\documentclass[a4paper,11pt,english]{article}

\usepackage[space]{cite}
\usepackage{epsfig}
\usepackage{babel}
\usepackage{a4wide}
\usepackage{amsmath}
\usepackage{amsthm}
\usepackage{amssymb}
\usepackage{units}
\usepackage{booktabs}
\usepackage{setspace}
\usepackage[a4paper,left=16mm,right=16mm, top=3.5cm, bottom=3cm]{geometry}
\usepackage{ifthen}
\usepackage{color}
\usepackage{framed}
\usepackage{bm}
\usepackage{enumitem}

\usepackage{nameref}

\makeatletter
\renewcommand{\@biblabel}[1]{\quad#1.}
\makeatother

\newcommand{\legi}[1]{{#1}}
\newcommand{\rokni}[1]{{#1}}
\newcommand{\david}[1]{{#1}}

\title{{\LARGE
A dynamic connectome supports the emergence of stable computational function of neural circuits through reward-based learning
\\[5mm]}
}

\author{
David Kappel$^{1,*}$, Robert Legenstein$^{*}$, Stefan Habenschuss, Michael Hsieh and Wolfgang Maass \\
Institute for Theoretical Computer Science\\
Graz University of Technology\\
8010 Graz, Austria \\
$^{1}$corresponding author: \texttt{kappel@igi.tugraz.at} \\[3mm]
$^{*}$equal contribution
}

\renewcommand{\d}{d}
\newcommand{\ve}[1]{\boldsymbol{#1}}

\newcommand{\expect}[2][]{\left\langle\, #2\, \right\rangle_{#1} }

\newtheorem{thm}{Theorem}

\newcommand{\cisd}{$\text{CI}_\text{SD}$}

\definecolor{gray}{gray}{0.2}



\newcommand{\cprob}[2]{p\left({#1}\,\middle\vert\,{#2}\right)}
\newcommand{\wiener}{\mathcal{W}}

\newcommand{\ddt}{\frac{\partial}{\partial t}}
\newcommand{\ddthetai}{\frac{\partial}{\partial \theta_i}}

\newcommand{\ddthetaisq}{\frac{\partial^2}{\partial \theta_{i}^2}}

\newcommand{\dd}[1]{\frac{\partial}{\partial #1}}

\newcommand{\hz}{y}

\newcommand{\bbz}{\mathbf{z}}
\newcommand{\bbzt}{\mathbf{z}(t)}

\newcommand{\bbztau}{\mathbf{z}(\tau)}

\newcommand{\bth}{\ve \theta}
\newcommand{\btht}{\ve \theta(t)}
\newcommand{\thi}{\theta_i(t)}

\newcommand{\pn}[2]{p_{\mathcal{N}}\left({#1}\,\middle\vert\,{#2}\right)}

\newcommand{\ps}[1]{p_{\mathcal{S}}\left({#1}\right)}

\newcommand{\syn}[1]{\text{\textsc{syn}}_{#1}}
\newcommand{\prei}{\text{\tiny{\textsc{pre}}}_{i}}
\newcommand{\posti}{\text{\tiny{\textsc{post}}}_{i}}
\newcommand{\preidef}{\text{\textsc{pre}}_{i}}
\newcommand{\postidef}{\text{\textsc{post}}_{i}}

\newcommand{\pl}[1]{{(\bf\MakeLowercase{#1})}}
\newcommand{\pr}[1]{{\bf\MakeLowercase{#1}}}
\newcommand{\figref}[2][]{{Fig.~\ref{#2}\MakeLowercase{#1}}}
\newcommand{\zk}{z_{\posti}}
\newcommand{\zkt}{\zk(t)}
\newcommand{\zks}{\zk(s)}

\newcommand{\fk}{f_{\posti}}
\newcommand{\fkt}{\fk(t)}
\newcommand{\fks}{\fk(s)}

\newcommand{\rt}{r(t)}
\newcommand{\rtau}{r(\tau)}
\newcommand{\rhatt}{\hat{r}(t)}
\newcommand{\rseq}{\ve r}
\newcommand{\rbin}{v_{b}}

\begin{document}

\maketitle

\begin{abstract}
Synaptic connections between neurons in the brain are dynamic because of continuously ongoing spine dynamics, axonal sprouting, and other processes. In fact, it was recently shown that the spontaneous synapse-autonomous component of spine dynamics is at least as large as the component that depends on the history of pre- and postsynaptic neural activity. These data are inconsistent with common models for network plasticity, and raise the questions how neural circuits can maintain a stable computational function in spite of these continuously ongoing processes, and what functional uses these ongoing processes might have. 
\rokni{Here, we present a rigorous theoretical framework for these seemingly stochastic spine dynamics and rewiring processes in the context of reward-based learning tasks.}
We show that spontaneous synapse-autonomous processes, in combination with reward signals such as dopamine, can explain the capability of networks of neurons in the brain to configure themselves for specific computational tasks, and to compensate automatically for later changes in the network or task. Furthermore we show theoretically and through computer simulations that stable computational performance is compatible with continuously ongoing synapse-autonomous changes. After reaching good computational performance it causes primarily a slow drift of network architecture and dynamics in task-irrelevant dimensions, as observed for neural activity in motor cortex and other areas. 
On the more abstract level of reinforcement learning the resulting model gives rise to an understanding of reward-driven network plasticity as \legi{continuous sampling of network configurations.}
(240 words)
\end{abstract}

\subsubsection*{Keywords}
Spine dynamics, rewiring, stochastic synaptic plasticity, synapse-autonomous processes, task-irrelevant dimensions in motor control, reward-modulated STDP, network computation, network plasticity, reinforcement learning, policy gradient, sampling.

\subsubsection*{Significance statement}

Networks of neurons in the brain do not have a fixed connectivity. We address the question how stable computational performance can be achieved by continuously changing neural circuits, and how these networks could even benefit from these changes. We show that the stationary distribution of network configurations provides a level of analysis where these issues can be addressed in a perspicuous manner. In particular, this theoretical framework allows us to address analytically the questions which rules for reward-gated synaptic rewiring and plasticity would work best in this context, and what impact different levels of activity-independent synaptic processes are likely to have. We demonstrate the viability of this approach through computer simulations and links to experimental data. (116 words)

\section*{Introduction}

The connectome is dynamic: Networks of neurons in the brain rewire themselves on a time scale of hours to days \cite{HoltmaatETAL:05,StettlerETAL:06,YangETAL:09,HoltmaatSvoboda:09,ZivAhissar:09,MinerbiETAL:09,KasaiETAL:10,LoewensteinETAL:11,LoewensteinETAL:15,RumpelTriesch:16,ChambersRumpel:17,vanOoyenButz-Ostendorf:17}. \david{This synaptic rewiring manifests in the emergence and vanishing of dendritic spines \cite{HoltmaatSvoboda:09}. Additional structural changes of established synapses are observable as a growth and shrinking of spine heads which take place even in the absence of neural activity \cite{YasumatsuETAL:08}.} The recent study of \cite{DvorkinZiv:16}, which includes in Fig.~8 a reanalysis of mouse brain data from \cite{KasthuriETAL:15}, showed that this spontaneous component is surprisingly large, at least as large as the impact of pre- and postsynaptic neural activity. \david{In addition, Nagaoka and colleagues provide direct evidence in vivo that the baseline turnover of dendritic spines is mediated by activity-independent intrinsic dynamics \cite{NagaokaETAL:16}. Furthermore, experimental data also suggest that task-dependent self-configuration of neural circuits is mediated by reward signals in \cite{YagishitaETAL:14}.
}

Other experimental data show that not only the connectome, but also the dynamics and function of neural circuits is subject to continuously ongoing changes. Continuously ongoing drifts of neural codes were reported in \cite{ZivETAL:13,DriscollHarvey:16}. Further data show that the mapping of inputs to outputs by neural networks that plan and control motor behavior are subject to a random walk on a slow time-scale of minutes to days,
that is conjectured to be related to stochastic synaptic rewiring and plasticity \rokni{\cite{van2013random,ChaisanguanthumETAL:14,RokniETAL:07}}. 

We address two questions that are raised by these data:
\begin{enumerate}[label=\roman*)]
\item How can stable network performance be achieved in spite of the experimentally found continuously ongoing rewiring and activity-independent synaptic plasticity in neural circuits?
\item What could be a functional role of these processes?
\end{enumerate}
Similar as \rokni{\cite{StatmanETAL:14,LoewensteinETAL:15,RokniETAL:07}} we model spontaneous synapse-autonomous spine dynamics of each potential synaptic connection $i$ through a stochastic process that modulates a corresponding parameter $\theta_i$. \rokni{We provide in this article a rigorous mathematical framework for such stochastic spine dynamics and rewiring processes. Our analysis shows that one} 
can describe the network configuration, i.e., the current state of the dynamic connectome and the strengths of all currently functional synapses, at any time point by a vector $\bth$ that encodes the current values $\theta_i$ for all potential synaptic connections $i$. The stochastic dynamics of this high-dimensional vector $\bth$ defines a Markov chain whose stationary distribution (illustrated in \figref[D]{fig:model-illustration}) provides insight into questions that address the relation between properties of local synaptic processes and the computational function of a neural network.

\rokni{Based on the well-studied paradigm for reward-based learning in neural networks, we} propose the following answer to question i): As long as most of the mass of this stationary distribution lies in regions or low-dimensional manifolds of the parameter space that produce good performance, stable network performance can be assured in spite of continuously ongoing movement of $\bth$ \rokni{\cite{LoewensteinETAL:15}}. \rokni{Our experimental results suggest that when a computational task has been learnt, most of the subsequent dynamics of $\bth$ takes place in task-irrelevant dimensions.}

\david{
The same model also provides an answer to question ii): Synapse-autonomous stochastic dynamics of the parameter vector $\bth$ enables the network not only to find in a high-dimensional regions with good network performance, but also to rewire the network in order to compensate for changes in the task. We analyze how the strength of the stochastic component of synaptic dynamics affects this compensation capability. We arrive at the conclusion that compensation works best for the task considered here if the stochastic component is as large as in experimental data \cite{DvorkinZiv:16}.
}

On the more abstract level of reinforcement learning, our theoretical framework for reward-driven network plasticity suggests a new algorithmic paradigm for network learning: \legi{policy sampling}. Compared with the familiar policy gradient learning \cite{Williams92simple,BaxterBartlett:00,PetersSchaal:06} this paradigm is more consistent with experimental data that suggest a continuously ongoing drift of network parameters. 

The resulting model for reward-gated network plasticity builds on the approach from \cite{KappelETAL:15} for unsupervised learning, that was only applicable to a specific neuron model and a specific STDP-rule.  
Since the new approach can be applied to arbitrary neuron models, in particular also to large data-based models of neural circuits and systems, it can be used to explore how data-based models for neural circuits and brain areas can attain and maintain a computational function.
(668 words)

\section*{Results}

\begin{figure}
\begin{center}
  \includegraphics{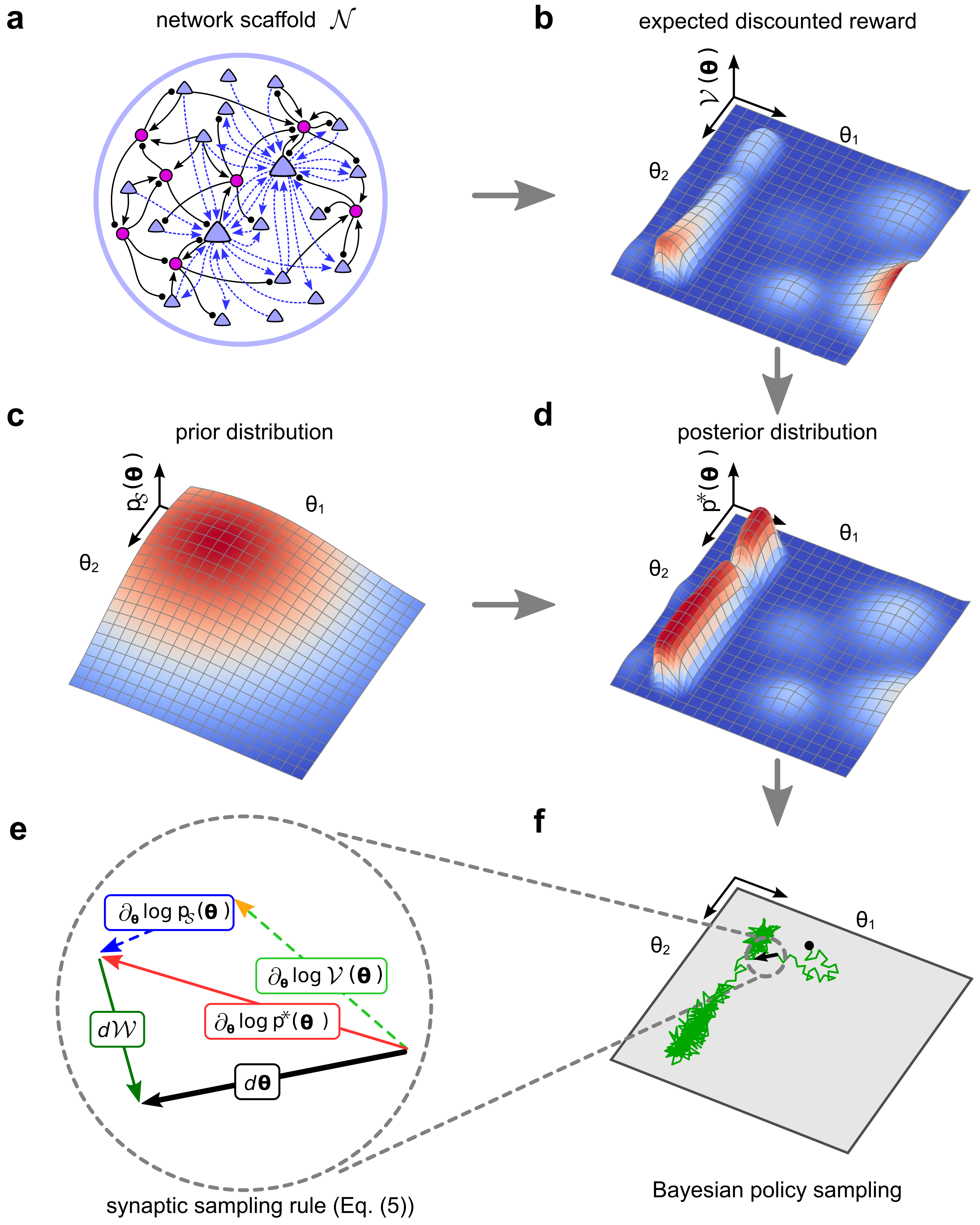}
\end{center}
\caption{{\bf Illustration of the theoretical framework.}
\pl{A} A neural network scaffold $\mathcal{N}$ of excitatory (blue triangles) and inhibitory (purple circles) neurons. Potential synaptic connections (broken blue arrows) of only two excitatory neurons are shown to keep the figure uncluttered. Synaptic connections (black connections) from and to inhibitory neurons are assumed to be fixed for simplicity.
\pl{B}
A reward landscape for two parameters $\ve\theta=\{\theta_1,\theta_2\}$ with several local optima. \rokni{Z-amplitude and color indicate the expected reward $\mathcal{V}(\bth)$ for given parameters $\bth$ (X-Y plane).}
\pl{C}
Example prior that prefers small values for $\theta_1$ and $\theta_2$.
\pl{D}
The posterior distribution $p^*(\ve \theta)$  that results as product of the prior from panel \pr{C} and the expected discounted reward of panel \pr{B}.
\pl{E}
Illustration of the dynamic forces (plasticity rule Eq.~\eqref{eq:sde}) that act on $\ve \theta$ in each sampling step $d \ve \theta$ (black) while sampling from the posterior distribution. The deterministic term (red), which consists of the first two terms (prior and reward expectation) in Eq.~\eqref{eq:sde}, is directed to the next local maximum of the posterior.
The stochastic term $d \wiener$ (green) of Eq.~\eqref{eq:sde} has a random direction.
\pl{F}
A single trajectory of \legi{policy sampling} from the posterior distribution of panel \pr{D} under Eq.~\eqref{eq:sde}, starting at the black dot. The parameter vector $\ve \theta$ fluctuates between different solutions, and moves primarily along the task-irrelevant dimension $\theta_2$.
}
\label{fig:model-illustration}
\end{figure}

We first address the design of a suitable theoretical framework for investigating the self-organization of neural circuits for specific computational tasks in the presence of spontaneous synapse-autonomous processes and rewards. There exist well-established models for reward-modulated synaptic plasticity, see e.g. \cite{FremauxETAL:10}, where reward signals gate common rules for synaptic plasticity, such as STDP. But these rules are lacking two components that we need here:
\begin{itemize}
  \item an integration of rewiring with plasticity rules that govern the modulation of the strengths of already existing synaptic connections 
  \item a term that reflects the spontaneous synapse-autonomous component of synaptic plasticity and rewiring.
\end{itemize}
In order to illustrate our approach we consider a neural network scaffold \rokni{(see \figref[A]{fig:model-illustration})} with a large number of potential synaptic connections between excitatory neurons. Only a subset of these potential connections is assumed to be functional at any point in time.

\rokni{
If one allows rewiring then the concept of a neural network becomes problematic, since the definition of a neural network typically includes its synaptic connections. Hence we refer to the set of neurons of a network, its set of potential synaptic connections, and its set of definite synaptic connections -- such as in our case connections from and to inhibitory neurons (see \figref[A]{fig:model-illustration}) -- as a \emph{network~ scaffold}. A network scaffold $\mathcal{N}$ together with a  parameter vector $\bth$ that specifies a particular selection of functional synaptic connections out of the set of potential connections and particular synaptic weights for these defines a concrete neural network, to which we also refer as \emph{network configuration}.
}

For simplicity we assume that only excitatory connections are plastic, but the model can be easily extended to also reflect plasticity of inhibitory synapses. For each potential synaptic connection $i$, we introduce a parameter $\theta_i$ that describes its state both for the case when this potential connection $i$ is currently not functional (this is the case when $\theta_i  \leq 0$) and when it is functional (i.e., $\theta_i > 0$). More precisely, $\theta_i$ encodes the current strength or weight $w_i$ of this synaptic connection through the formula
\begin{equation}
 w_{i} \;=\; \begin{cases} 
                 \exp( \theta_i - \theta_0 ) \quad & \text{if } \theta_i > 0 \qquad (\textit{functional synaptic connection}) \\
                 0 \quad & \text{if } \theta_i \leq 0  \qquad (\textit{non-functional potential connection}) \\
               \end{cases} \;,
               \label{eq:thetamap}
\end{equation}
with a positive offset parameter $\theta_0$ that regulates the initial strength of new functional synaptic connections (we set $\theta_0 =3$ in our simulations).

The exponential function in Eq.~\eqref{eq:thetamap} turns out to be useful for relating the dynamics of $\theta_i$ to experimental data on the dynamics of synaptic weights. The volume -- or image brightness in Ca-imaging -- of a dendritic spine is commonly assumed to be proportional to the strength $w_i$ of a synapse \cite{HoltmaatETAL:05}. The logarithm of this estimate for $w_i$ was shown in Fig.~2i of \cite{HoltmaatETAL:06} and also in \cite{YasumatsuETAL:08,LoewensteinETAL:11} to exhibit a dynamics similar to that of an Ornstein-Uhlenbeck process, i.e., a random walk in conjunction with a force that draws the random walk back to its initial state. Hence if $\theta_i$  is chosen to be proportional to the logarithm of $w_i$, it is justified to model the spontaneous dynamics of $\theta_i$ as an Ornstein-Uhlenbeck process. This is done in our model, as we will explain after Eq.~\eqref{eq:sde} and demonstrate in Fig.~\ref{fig:pattern-classification}c. The logarithmic transformation also ensures that additive increments of $\theta_i$ yield multiplicative updates of $w_i$, which have been observed experimentally \cite{LoewensteinETAL:11}.

Altogether our model needs to create a dynamics for $\theta_i$ that is not only consistent with experimental data on spontaneous spine dynamics, but is for the case $\theta_i > 0$ also consistent with rules for reward-modulated synaptic plasticity as in \cite{FremauxETAL:10}.  This suggests to look for plasticity rules of the form
\begin{equation}
	d \theta_i \;=\; \beta \, \times \,(\text{deterministic plasticity rule}) \times dt \;+\; \sqrt{2 \beta T} \, d\mathcal{W}_i \; ,
	\label{eq:sde-illustration}
\end{equation}
where the deterministic plasticity rule could for example be a standard reward-based plasticity rule. 
\rokni{We will argue below that it makes sense to include also an activity-independent prior in this deterministic component of rule \eqref{eq:sde-illustration},} both for functional reasons and in order to fit data on spontaneous spine dynamics. \legi{We will further see that when the activity-independent prior dominates, we obtain the Ornstein-Uhlenbeck process mentioned above.}
The stochastic term $d\mathcal{W}_i$ in Eq.~\eqref{eq:sde-illustration} is an infinitesimal step of a random walk, more precisely for a Wiener process $\mathcal{W}_i$. A Wiener process is a standard model for Brownian motion in one dimension \cite{Gardiner:04}. The term $\sqrt{2 \beta T}$ scales the strength of this stochastic component in terms of a ``temperature'' $T$ and a learning rate $\beta$, and is chosen to be of a form that supports analogies to statistical physics. The presence of this stochastic term makes it unrealistic to expect that $\theta_i$ converges to a particular value under the dynamics defined by Eq.~\eqref{eq:sde-illustration}. In fact, in contrast to many standard differential equations, the stochastic differential equation or SDE \eqref{eq:sde-illustration} does not have a single trajectory of $\theta_i$ as solution, but an infinite family of trajectories that result from different random walks. 

We propose to focus -- instead of the common analysis of the convergence of weights to specific values as invariants -- on the most prominent invariant that a stochastic process can offer: the longterm stationary distribution of synaptic connections and weights. The stationary distribution of the vector $\bth$ of all synaptic parameters $\theta_i$ informs us about the statistics of the infinitely many different solutions of a stochastic differential equation of the form \eqref{eq:sde-illustration}. In particular, it informs us about the fraction of time at which particular values of $\bth$ will be visited by these solutions \david{(see \nameref{sec:methods} for details)}. We show that a large class of reward-based plasticity rules produce in the context of an equation of the form \eqref{eq:sde-illustration} a stationary distribution of $\bth$ that can be clearly related to reward expectation for the neural network, and hence to its computational function.  

We want to address the question \legi{whether} reward-based plasticity rules achieve in the context with other
terms in Eq.~\eqref{eq:sde-illustration} that the resulting stationary distribution of network configurations has most of its mass on
highly rewarded network configurations. A key observation is that if the first term on the right-hand-side
of \eqref{eq:sde-illustration} can be written for all potential synaptic connections $i$ in the form $\ddthetai \log p^{*}(\bth)$, where $p^{*}(\bth)$ is some arbitrary given distribution and $\ddthetai$ denotes the partial derivative with respect to parameter $\theta_i$, then these
stochastic processes
\begin{equation}
d \theta_i \;=\; \beta \, \ddthetai \,\log p^*(\bth) \, dt  \;+ \; \sqrt{2 \beta T} \, d \wiener_{i} 
\label{eq:sde-reduced}
\end{equation}
give rise to a stationary distribution that is proportional to $p^{*}(\bth)^\frac{1}{T}$.
Hence,  a rule for reward-based synaptic plasticity that can be written in the form $\ddthetai \log p^{*}(\bth)$, where $p^{*}(\bth)$ has most of its mass on highly rewarded network configurations $\bth$, achieves that the network will spend most of its time in highly rewarded network configurations. This will hold even if the network does not converge to or stay in any particular network configuration $\bth$ (see \figref[D,F]{fig:model-illustration} for an illustration). Furthermore the role of the temperature $T$ in \eqref{eq:sde-reduced} becomes clearly visible in this result: if $T$ is large the resulting stationary distribution flattens the distribution $p^{*}(\bth)$, whereas for $0<T<1$ the network will remain for larger fractions of the time in those regions of the parameter space where $p^{*}(\bth)$ achieves its largest values. In fact, if the temperature $T$ converges to 0, the resulting stationary distribution degenerates to one that has all of its mass on the network configuration $\bth$ for which  $p^{*}(\bth)$ reaches its global maximum, as in simulated annealing \cite{KirkpatrickETAL:83}.

We will focus on target distributions $p^{*}(\bth)$ of the form 
\begin{equation}
  p^*(\bth) \;\propto\; p_S(\bth) \, \times \, \mathcal{V}(\bth)  \; ,
\label{eq:bayes}
\end{equation}
where $\propto$ denotes proportionality up to a positive normalizing constant. $p_S(\bth)$ can encode structural priors of the network scaffold $\mathcal{N}$. For example, it can encode a preference for sparsely connected networks. This happens when $p_S(\bth)$ has most of its mass near $\boldsymbol{0}$, see \figref[C]{fig:model-illustration} for an illustration. But it could also convey genetically encoded or previously learnt information, such as a preference for having strong synaptic connections between two specific populations of neurons. The term $\mathcal{V}(\bth)$ in Eq.~\eqref{eq:bayes} denotes the expected discounted reward associated with a given parameter vector $\bth$ (see \figref[B]{fig:model-illustration}). Eq.~\eqref{eq:sde-reduced} for the stochastic dynamics of parameters takes then the form
\begin{equation}
d \theta_i \;=\; \beta \, \left( \ddthetai \,\log \ps{\bth}  \, + \, \ddthetai \log \mathcal{V}(\bth)  \right)  dt  \;+ \; \sqrt{2 \beta T} \, d \wiener_{i} \;.
\label{eq:sde}
\end{equation}
When the term $\ddthetai \log \mathcal{V}(\bth)$ vanishes, this equation models spontaneous spine dynamics. We will make sure that this term vanishes for all potential synaptic connections $i$ that are currently not functional, i.e., where $\theta_{i} \leq 0$. If one chooses a Gaussian distribution as prior $\ps{\bth}$, the dynamics of \eqref{eq:sde} amounts in the case $ \ddthetai \,\log \mathcal{V}({\bth}) = 0 $ to an Ornstein-Uhlenbeck process. \david{There is currently no generally accepted description of spine dynamics. Ornstein-Uhlenbeck dynamics has previously been proposed as a simple model for experimentally observed spontaneous spine dynamics \cite{YasumatsuETAL:08,LoewensteinETAL:11,LoewensteinETAL:15}. Another proposed model uses a combination of multiplicative and additive stochastic dynamics \cite{StatmanETAL:14,RubinskiZiv:15}. We used in our simulations a Gaussian distribution that prefers small but nonzero weights for the prior $\ps{\bth}$. Hence, our model  \eqref{eq:sde} is consistent with previous Ornstein-Uhlenbeck models for spontaneous spine dynamics.}

Thus altogether we arrive at a model for the interaction of stochastic spine dynamics with reward where the usually considered deterministic convergence to network configurations $\bth$ that represent  local maxima of expected reward $\mathcal{V}(\bth)$ (e.g. to the local maxima in \figref[B]{fig:model-illustration}) is replaced by a stochastic model. If the stochastic dynamics of $\bth$ is defined by local stochastic processes of the form \eqref{eq:sde}, as indicated in \figref[E]{fig:model-illustration}, the resulting stochastic model for network plasticity will spend most of its time in network configurations $\bth$ where the posterior $p^{*}(\bth)$, illustrated in \figref[D]{fig:model-illustration}, approximately reaches its maximal value. This provides on the statistical level a guarantee of \david{task performance}, in spite of ongoing stochastic dynamics of all the parameters $\theta_i$.

\subsection*{Reward-based rewiring and synaptic plasticity as \legi{policy sampling}}

\begin{figure}
\begin{center}
  \includegraphics{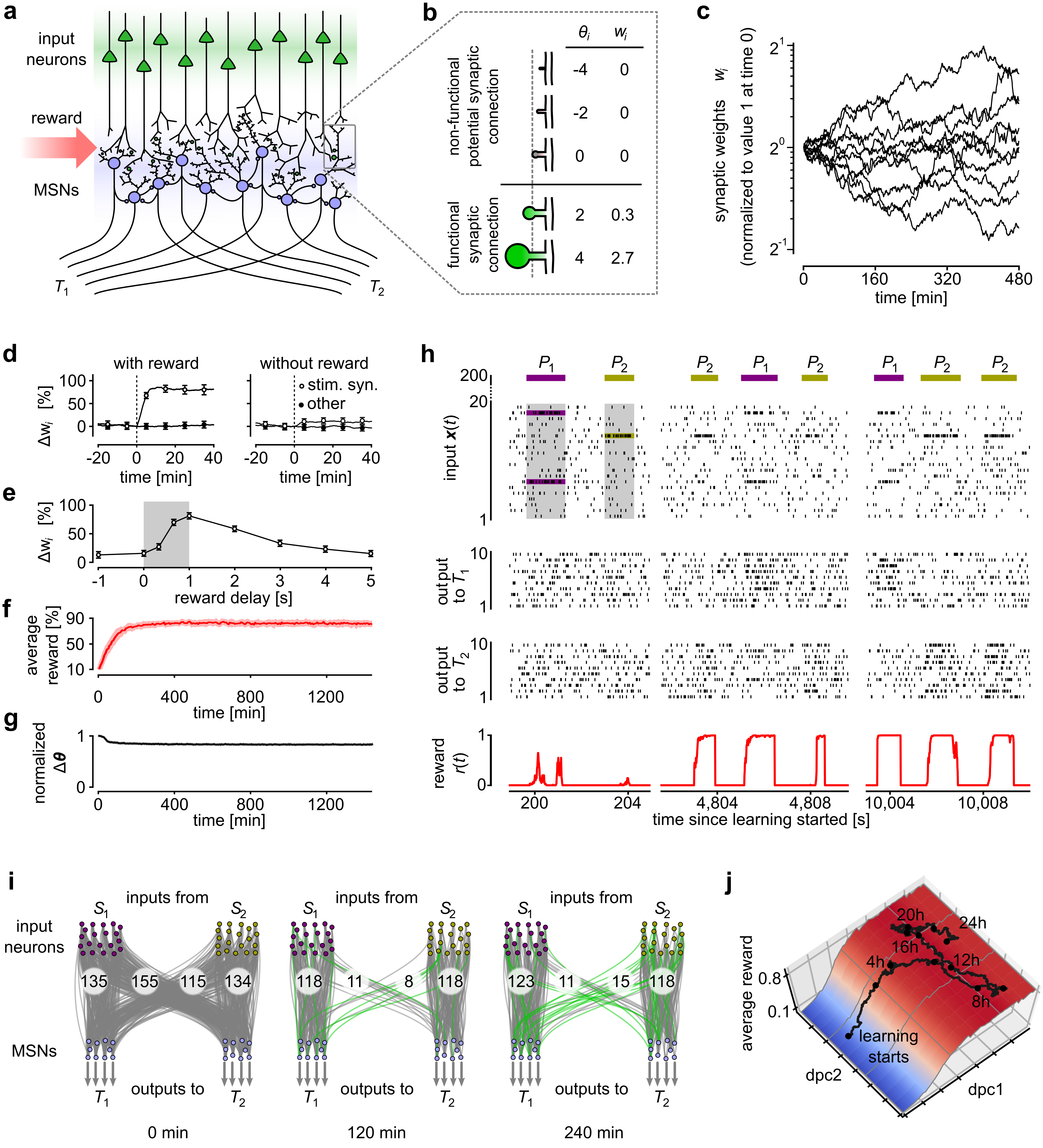}
\end{center}
\caption{{\bf Reward-based routing of input patterns.}
\pl{A} Illustration of the  network scaffold. A population of 20 model MSNs (blue) receives input from 200 excitatory input neurons (green) that model cortical neurons. Potential synaptic connections between these 2 populations of neurons were subject to reward-based synaptic sampling. In addition, fixed lateral connections provided recurrent inhibitory input to the MSNs.
\textit{Caption continued on next page...}
}
\label{fig:pattern-classification}
\end{figure}

\begin{figure}[htp]
\textit{Caption of \figref{fig:pattern-classification} continued:}
The MSNs were divided into two groups, each projecting exclusively to one of two target areas $T_1$ and $T_2$. Reward was delivered whenever the network managed to route an input pattern $P_i$ primarily to that group of MSNs that projected to target area $T_i$.
\pl{B}
Illustration of the model for spine dynamics. Five potential synaptic connections at different states are shown. Synaptic spines are represented by circular volumes with diameters proportional to $\sqrt[3]{w_i}$  for functional connections, assuming a linear correlation between spine-head volume and synaptic efficacy $w_i$ \cite{MatsuzakiETAL:01}. 
\pl{C}
Dynamics of weights $w_i$ in log-scale for 10 potential synaptic connections $i$ when the activity-dependent term $\ddthetai \log \mathcal{V}(\bth)  dt$ in Eq.~\eqref{eq:sde} is set equal to zero). Consistent with experimental date (see e.g. Fig. 2i of \cite{HoltmaatETAL:06}) the dynamics is in this case consistent with an Ornstein-Uhlenbeck process in the logarithmic scale. Weight values are plotted relative to the initial value at time 0.
\pl{D, E}
Dynamics of a model synapse when a reward-modulated STDP pairing protocol as in \cite{YagishitaETAL:14} was applied.
\pl{D} Reward delivery after repeated firing of the presynaptic neuron before the postsynaptic neuron resulted in a strong weight increase (left). This effect was reduced without reward (right), and prevented completely if no presynaptic stimulus was applied. Values in \pr{D} and \pr{E} represent percentage of weight changes relative the pairing onset time (dashed line, means and s.e.m. over 50 synapses).
 Compare with Fig.~1F,G in \cite{YagishitaETAL:14}.
\pl{E}
Dependence of resulting changes in synaptic weights in our model as a function of the delay of reward delivery. Gray shaded rectangle indicates the time window of STDP pairing application. Reward delays denote time between paring and reward onset. Compare to Figure 1O in \cite{YagishitaETAL:14}.
\pl{F}
The average reward achieved by the network increased quickly during learning according to Eq.~\eqref{eq:sde} (mean over 5 independent trial runs; shaded area indicates s.e.m.).
\pl{G}
Synaptic parameters kept changing throughout the experiment in \pr{F}. The magnitude of the change of the synaptic parameter vector $\bth$ is shown (mean $\pm$ s.e.m. as in \pr{F}; Euclidean norm, normalized to the maximum value). The parameter change peaks at the onset of learning, but remains high (larger than $80\%$ of the maximum value) even when stable performance has been reached.
\pl{H}
Spiking activity of the network during learning. Activities of 20 randomly selected input neurons and all MSNs are shown. 3 salient input neurons (belonging to pools $S_1$ or $S_2$ in \pr{I}) are highlighted. Most neurons have learnt to fire at a higher rate for the input pattern $P_j$ that corresponds to the target area $T_j$ to which they are projecting.
Bottom: reward delivered to the network.
\pl{I}
Dynamics of network rewiring throughout learning. Snapshots of network configurations for the times $t$ indicated below the plots are shown. Gray lines indicate active connections between neurons; connections that were not present at the preceding snapshot are highlighted in green. All output neurons and two subsets of input neurons that fire strongly in pattern $P_1$ or $P_2$ are shown (pools $S_1$ and $S_2$, 20 neurons each). Numbers denote total counts of functional connections between pools. The connectivity was initially dense and then rapidly restructured and became sparser. Rewiring took place all the time throughout learning.
\rokni{
\pl{J}
Analysis of random exploration in task-irrelevant dimensions of the parameter space. Projection of the parameter vector $\bth$ to the two dPCA components that best explain the variance of the average reward. dpc1 explains more than 99.9\% of the reward variance (dpc2 and higher dimensions less than 0.1\%). A single trajectory of the high-dimensional synaptic parameter vector over 24 hours of learning projected onto dpc1 and dpc2 is shown. Amplitude on the y-axis denotes the estimated average reward (in fractions of the total maximum achievable reward). After converging to a region of high reward (movement mainly along dpc1) network continues to explore task-irrelevant dimensions (movement mainly along dpc2).
}
\end{figure}

We assume that all \david{synapses and neurons in the network scaffold} $\mathcal{N}$ receives reward signals $\rt$ at certain times $t$, corresponding for example to dopamine signals in the brain (see \cite{CollinsFrank:16} for a recent discussion of related experimental data). The expected discounted reward $\mathcal{V}(\bth)$ that occurs in the second term of Eq.~\eqref{eq:sde} is the \david{expectation of the time integral over} all future rewards $\rt$, while discounting more remote rewards exponentially, see Eq.~\eqref{eqn:reward-prob-factorized} in \nameref{sec:methods}. \figref[B]{fig:model-illustration} shows a hypothetical $\mathcal{V}(\bth)$-landscape over two parameters $\theta_1, \theta_2$. The posterior $p^*(\bth)$ shown in \figref[D]{fig:model-illustration} is then proportional to the product of $\mathcal{V}(\bth)$ (panel {b) and the prior (panel {c}).

The computational behavior of the network configuration, i.e., the mapping of network inputs to network outputs that is encoded by the parameter vector $\bth$, is referred to as a policy in the context of reinforcement learning theory. \david{
The parameters $\bth$ (and therefore the policy) are gradually changed through Eq.~\eqref{eq:sde} such that the expected discounted reward $\mathcal{V}(\bth)$ is increased: The parameter dynamics follows the gradient of $\log \mathcal{V}(\bth)$, i.e., $\frac{d \theta_i}{d t} = \beta \frac{\partial}{\partial \theta_i} \log \mathcal{V}(\bth)$, where $\beta>0$ is a small learning rate. When the parameter dynamics is given solely by the second term in the parenthesis of Eq.~\eqref{eq:sde}, $\ddthetai \log \mathcal{V}(\bth)$, we recover for the case $\theta_i > 0$ deterministic policy gradient learning \cite{Williams92simple,BaxterBartlett:00,PetersSchaal:06}.}

\david{For a network scaffold  $\mathcal{N}$ of spiking neurons,} the derivative $\frac{\partial}{\partial \theta_i} \log \mathcal{V}(\bth)$  gives rise to synaptic updates at a synapse $i$ that are essentially given by the product of the current reward signal $\rt$ and an eligibility trace that depends on pre- or postsynaptic firing times, see {\em Synaptic dynamics for the reward-based synaptic sampling model} in \nameref{sec:methods}. Such plasticity rules 
have previously been proposed by \cite{Seung:03,XieSeung:04,Izhikevich:07,PfisterETAL:06,Florian:07,LegensteinETAL:08,UrbanczikSenn:09}. For non-spiking neural networks, a similar update rule was first introduced by Williams and termed the REINFORCE rule \cite{Williams92simple}. 

In contrast to policy gradient, reinforcement learning in the presence of the stochastic last term in Eq.~\eqref{eq:sde} cannot converge to any network configuration. Instead, the dynamics of Eq.~\eqref{eq:sde} produces continuously changing network configurations, with a preference for configurations that both satisfy constraints from the prior $p_S (\bth)$  and provide a large expected reward $\mathcal{V}(\bth)$, see \figref[D,F]{fig:model-illustration}. Hence this type of reinforcement learning samples continuously from a posterior distribution of network configurations. This is rigorously proven in Theorem~\ref{lem:single_sup} of \nameref{sec:methods}. We refer to this reinforcement learning model as \legi{{\em policy sampling}}, and to the family of reward-based plasticity rules that are defined by Eq.~\eqref{eq:sde} as {\em reward-based synaptic sampling}.

Another key difference to previous models for reward-gated synaptic plasticity and policy gradient learning is, apart from the stochastic last term of Eq.~\eqref{eq:sde}, that the deterministic first term of Eq.~\eqref{eq:sde} also contains a reward-independent component $\frac{\partial}{\partial \theta_i} \log p_S (\bth)$ that arises from a prior $p_S (\bth)$ for network configurations. In our simulations we consider a simple Gaussian prior $p_S (\bth)$ with mean $\boldsymbol{0}$ that encodes a preference for sparse connectivity (see Eq.~\eqref{eqn:dprior-dtheta}).

It is important that the dynamics of disconnected synapses, i.e., of synapses $i$ with $\theta_i \leq 0$ or equivalently $w_i = 0$, does not depend on \david{pre-/postsynaptic neural activity or reward} since non-functional synapses do not have access to such information. This is automatically achieved through our ansatz $\frac{\partial}{\partial \theta_i} \log \mathcal{V}(\bth)$ for the reward-dependent component in Eq.~\eqref{eq:sde}, since a simple derivation shows that it entails that the factor $w_i$ appears in front of the term that depends on pre- and postsynaptic activity, see Eq.~\eqref{eqn:eligibility-trace}. Instead, the dynamics of $\theta_i$ depends for $\theta_i \leq 0 $ only on the prior and the stochastic term $d \mathcal{W}_i$. This results in a distribution over waiting times between downwards and upwards crossing of the threshold $\theta_i = 0 $ that was found to be similar to the distribution of inter-event times of a Poisson point process, see \cite{DingRangarajan:04} for a detailed analysis. This theoretical result suggest a simple approximation of the dynamics of Eq.~\eqref{eq:sde} for currently non-functional synaptic connections, where the process \eqref{eq:sde} is suspended whenever $\theta_i$ becomes negative, and continued with $\theta_i = 0$ after a waiting time that is drawn from an exponential distribution. As in \cite{DegerETAL:16} this can be realized by letting a non-functional synapse become functional at any discrete time step with some fixed probability (Poisson process). We have compared in \figref[C]{fig:exp-peters} the resulting learning dynamics of the network for this simple approximation with that of the process defined by Eq.~\eqref{eq:sde}.

\subsection*{Task-dependent routing of information through the interaction of stochastic spine dynamics with rewards}

Experimental evidence about gating of spine dynamics by reward signals in the form of dopamine is available for the synaptic connections from the cortex to the entrance stage of the basal ganglia, the medium spiny neurons (MSNs) in the striatum \cite{YagishitaETAL:14}. They report that the volumes of their dendritic spines show significant changes only when pre- and postsynaptic activity is paired with precisely timed delivery of dopamine (see  \cite{YagishitaETAL:14}, Fig.~1 E-G, O). More precisely, an STDP pairing protocol followed by dopamine uncaging induced strong LTP in synapses onto MSNs, whereas the same protocol without dopamine uncaging lead only to a minor increase of synaptic efficacies. 

MSNs can be viewed as readouts from a large number of cortical areas, that become specialized for particular motor functions, e.g. movements of the hand or leg. We asked whether reward gating of spine dynamics according to the experimental data of \cite{YagishitaETAL:14} can explain such task dependent specialization of MSNs. More concretely, we asked whether it can achieve that two different distributed activity patterns $P_1$, $P_2$ of upstream neurons in the cortex get routed to two different ensembles of MSNs, and thereby to two different downstream targets $T_1$ and $T_2$ of these MSNs (see \figref[A,H,I]{fig:pattern-classification}). We assumed that for each upstream activity pattern $P_j$ a particular subset $S_j$ of upstream neurons is most active, $j=1,2$. Hence this routing task amounted to routing synaptic input from $S_j$ to those MSNs that project to downstream neuron $T_j$.

We applied to all potential synaptic connections $i$ from upstream neurons to MSNs a learning rule according to Eq.~\eqref{eq:sde}, more precisely, the rule for reward-gated STDP (Eq.~\eqref{eqn:eligibility-trace}, Eq.~\eqref{eqn:gradient-est} and Eq.~\eqref{eqn:std-synapse}) that results from this general framework. The parameters of the model were adapted to qualitatively reproduce the results from Figures 1F,G of \cite{YagishitaETAL:14} when the same STDP protocol was applied to our model (see \figref[D,E]{fig:pattern-classification}). The parameter values are reported in Tab.~\ref{tab:parameters} in \nameref{sec:methods}. If not stated otherwise, we applied these parameters in all following experiments.
\rokni {
In {\em Role of the prior distribution} in \nameref{sec:methods}, we further analyze the impact of different prior distributions on task performance and network connectivity.
}

Our simple model consisted of 20 inhibitory model MSNs with lateral recurrent connections. These received excitatory input from 200 input neurons. The synapses from input neurons to model MSNs were subject to our plasticity rule. Multiple connections were allowed between each pair of input neuron and MSN (see \nameref{sec:methods}). The MSNs were randomly divided into two assemblies, each projecting exclusively to one of two downstream target areas $T_1$ and $T_2$.
Cortical input $\ve{x}(t)$ was modeled as Poisson spike trains from the 200 input neurons with instantaneous rates defined by two prototype rate patterns $P_1$ and $P_2$, see \figref[H]{fig:pattern-classification}. The task was to learn to activate $T_1$-projecting neurons and to silence $T_2$-projecting neurons whenever pattern $P_1$ was presented as cortical input. For pattern $P_2$, the activation should be reversed: activate $T_2$-projecting neurons and silence those projecting to $T_1$. This desired function was defined through a reward signal $\rt$ that was proportional to the ratio between the mean firing rate of MSNs projecting to the desired target and that of MSNs projecting to the non-desired target area (see \nameref{sec:methods}).

\figref[H]{fig:pattern-classification} shows the firing activity and reward signal of the network during segments of one simulation run. After about 80 minutes of simulated biological time, each group of MSNs had learned to increase its firing rate when the activity pattern $P_j$ associated with its projection target  $T_j$ was presented. \figref[F]{fig:pattern-classification} shows the average reward throughout learning. After 3 hours of learning about $82\%$ of the maximum reward was acquired on average, and this level was maintained during prolonged learning.

\figref[G]{fig:pattern-classification} shows that the parameter vector $\bth$ kept moving at almost the same speed even after a high plateau of rewards had been reached. Hence these ongoing parameter changes took place in dimensions that were irrelevant for the reward-level.

\figref[I]{fig:pattern-classification} provides snapshots of the underlying ``dynamic connectome'' \cite{RumpelTriesch:16} at different points of time. New synaptic connections that were not present at the preceding snapshot are colored green. One sees that the bulk of the connections maintained a solution of the task to route inputs from $S_1$ to target area $T_1$ and inputs from $S_2$ to target area $S_2$. But the identity of these connections, a task-irrelevant dimension, kept changing. In addition the network always maintained some connections to the currently undesired target area, thereby providing the basis for a swift built-up of these connections if these connections would suddenly also become rewarded.

\rokni{
We further examine the exploration along task-irrelevant dimensions in \figref[J]{fig:pattern-classification}. Here, the high-dimensional parameter vector over a training experiment of 24~h projected to the first two components of the demixed principal component analysis (dPCA) that best explain the variance of the average reward is shown (see Methods and \cite{KobakETAL:16}). The first component (dpc1) explains $>99.9\%$ of the variance. Movement of the parameter vector mainly takes place along this dimensions during the first 4 hours of learning. After the performance has converged to a high value, exploration continues along other components (dpc2, and higher components) that explain less than $0.1\%$ of the average reward variance.
}

This simulation experiment showed that reward-gated spine dynamics as analyzed in \cite{YagishitaETAL:14} is sufficiently powerful from the functional perspective to rewire networks so that each signal is delivered to its intended target.

\begin{figure}[hp]
\begin{center}
  \includegraphics{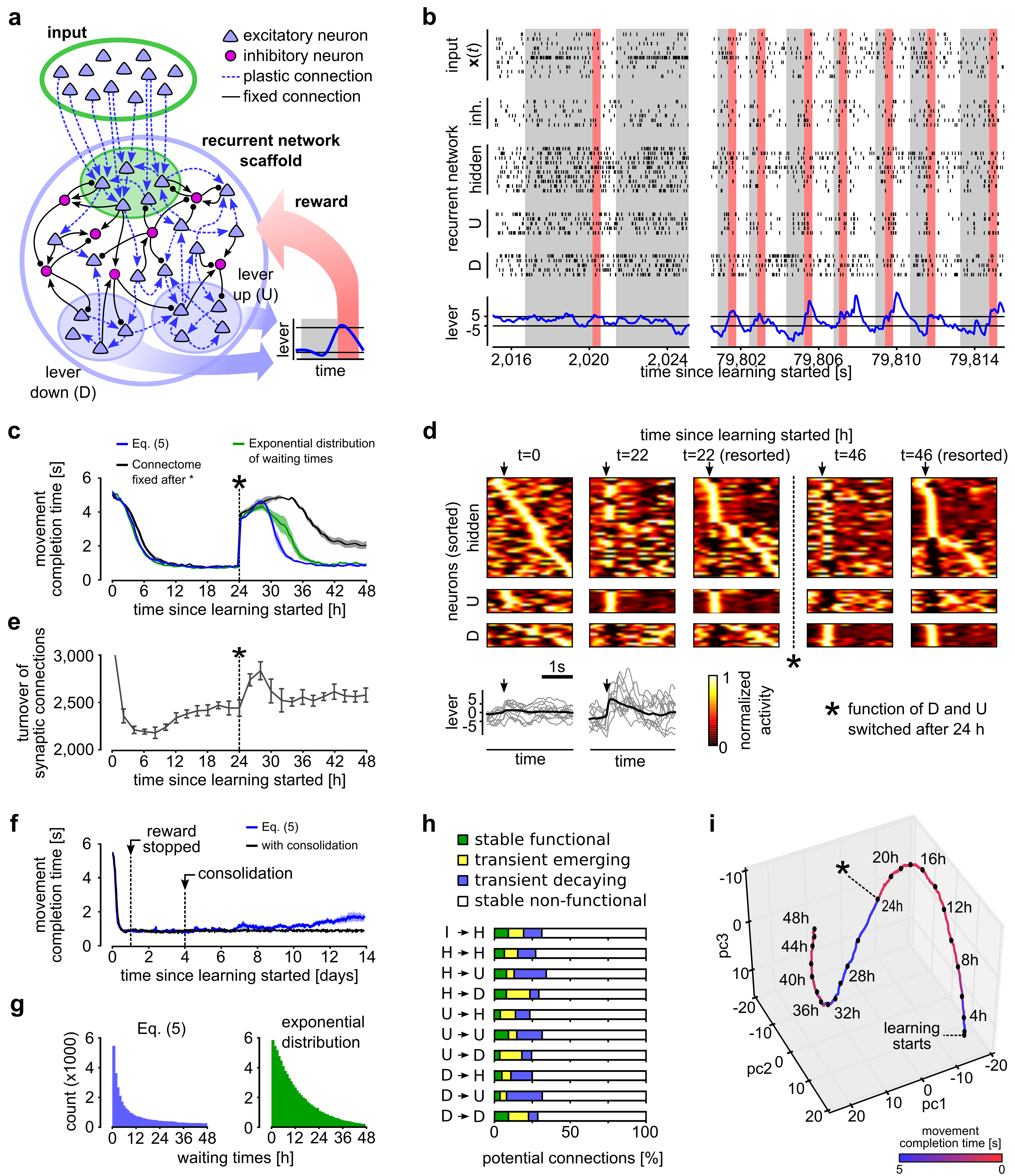}
\end{center}
\caption{{\bf Reward-based self-configuration and compensation capability of a recurrent neural network.}
\pl{A}
Network scaffold and task schematic. Symbol convention as in \figref[A]{fig:model-illustration}. A recurrent network scaffold of excitatory and inhibitory neurons (large blue circle); a subset of excitatory neurons received input from afferent excitatory neurons (indicated by green shading).
 \textit{Caption continued on next page...}} 
\label{fig:exp-peters}
\end{figure}

\begin{figure}[ht!]
\textit{Caption of \figref{fig:exp-peters} continued:}
From the remaining excitatory neurons, two pools D and U were randomly selected to control lever movement (blue shaded areas).
Bottom inset: stereotypical movement that had to be generated to receive a reward.
\pl{B}
Spiking activity of the network at learning onset and after 22 hours of learning. Activities of random subsets of neurons from all populations are shown (hidden: excitatory neurons of the recurrent network, which are not in pool D or U). Bottom: lever position inferred from the neural activity in pools D and U. Rewards are indicated by red bars. Gray shaded areas indicate cue presentation.
\pl{C}
Task performance quantified by the average time from cue presentation onset to movement completion. The network was able to solve this task in less than 1 seconds on average after about 8 hours of learning. A task change was introduced at time 24~h (asterisk; function of D and U switched), which was quickly compensated by the network. Using a simplified version of the learning rule, where the re-introduction of non-functional potential connections was approximated using exponentially distributed waiting times (green), yielded similar results (see also panel e). If the connectome was kept fixed after the task change at 24~h performance was significantly worse (black).
\pl{D}
Trial-averaged network activity (top) and lever movements (bottom). Activity traces are aligned to movement onsets (arrows). Y-axis of trial-averaged activity plots are sorted by the time of highest firing rate within the movement at various times during learning: sorting of the first and second plot is based on the activity at $t=0$ h, third and fourth by that at $t=22$ h, fifth is resorted by the activity at $t=46$ h. Network activity is clearly restructured through learning with particularly stereotypical assemblies for sharp upward movements. Bottom: average lever movement (black) and 10 individual movements (gray).
\pl{E}
Turnover of synaptic connections for the experiment shown in \pr{D}. Y-axis is clipped at 3000. Turnover rate during the first two hours was around 12.000 synapses ($\sim 25\%$) and then decreased rapidly. Another increase in spine turnover rate can be observed  after the task change at time 24 h.
\rokni{
\pl{F}
Effect of forgetting due to parameter diffusion over 14 simulated days. Application of reward was stopped after 24 hours when the network had learned to reliably solve the task. Parameters subsequently continue to evolve according to the SDE \eqref{eq:sde}. Onset of forgetting can be observed after day 6. A simple consolidation mechanism triggered after 4 days reliably prevents forgetting.
}
\pl{G}
Histograms of time intervals between disappearance and reappearance of synapses (waiting times) for the exact (upper plot) and approximate (lower plot) learning rule.
\pl{H}
Relative fraction of potential synaptic connections that were stably non-functional, transiently decaying, transiently emerging or stably function during the re-learning phase for the experiment shown in \pr{D}.
\pl{I}
PCA of a random subset of the parameters $\theta_i$. The plot suggests continuing dynamics in task-irrelevant dimensions after the learning goal has been reached (indicated by red color). When the function of the neuron pools U and D was switched after 24 h, the synaptic parameters migrated to a new region. All plots show means over 5 independent runs (error bars: s.e.m.).
\end{figure}

\subsection*{A model for task-dependent self-configuration of a recurrent network of excitatory and inhibitory spiking neurons}
 
We next asked, whether our simple integrated model for reward-modulated rewiring and synaptic plasticity of neural circuits according to Eq.~\eqref{eq:sde} could also explain the emergence of specific computations in recurrent networks of spiking neurons. As paradigm for a specific computational task we took a simplified version of the task that mice learned to carry out in the experimental setup of \cite{PetersETAL:14}. There a reward was given whenever a lever was pressed \david{within a given time window indicated by an auditory cue}. This task is particular suitable for our context, since spine turnover and changes of network activity were continuously monitored in \cite{PetersETAL:14} while the animals learned this task.

We adapted the learning task of \cite{PetersETAL:14} in the following way for our model (see \figref[A]{fig:exp-peters}).
The beginning of a trial was indicated through the presentation of a cue input pattern $\ve{x}(t)$: a fixed, randomly generated rate pattern for all 200 input neurons that lasted until the task was completed, but at most 10s. As network scaffold $\mathcal{N}$ we took a generic recurrent network of excitatory and inhibitory spiking neurons with connectivity parameters for connections between excitatory and inhibitory neurons according to data from layer 2/3 in mouse cortex \cite{AvermannETAL:12}. The network consisted of 60 excitatory and 20 inhibitory neurons (see \figref[A]{fig:exp-peters}). Half of the excitatory neurons could potentially receive synaptic connections from the 200 excitatory input neurons. From the remaining 30 neurons we randomly selected one pool D of 10 excitatory neurons to cause downwards movements of the lever, and another pool U of 10 neurons for upwards movements. We refer to the 40 excitatory neurons that were not members of D or U as hidden neurons. All excitatory synaptic connections from the external input (cue) and between the 60 excitatory neurons (including those in the pools D and U) in the network were subjected to reward-based synaptic sampling. 

\david{
To decode the lever position we filtered the population spikes of D and U with a smoothing kernel. The filtered population spikes of D were then subtracted from those of U to determine the lever position (seem \nameref{sec:methods} for details). When the lever position crossed the threshold +5 after first crossing a lower threshold -5 (black horizontal lines in \figref[A,B]{fig:exp-peters}) within 10~s after cue onset a 400~ms reward window was initiated during which $\rt$ was set to $1$ (red vertical bars in \figref[b]{fig:exp-peters}). Unsuccessful trials were aborted after 10 seconds and no reward was delivered. After each trial a brief holding phase of random length was inserted, during which input neurons were set to a background input rate of 2~Hz.
}

Thus, the network had to learn without any guidance, except for the reward in response to good performance, to create after the onset of the cue first higher firing in pool D, and then higher firing in pool U. This task was challenging, since the network had no information which neurons belonged to pools D and U. Moreover, the synapses did not ``know'' whether they connected to hidden neurons, neurons within a pool, hidden neurons and pool-neurons, or input neurons with other neurons. The plasticity of all these different synapses was gated by the same global reward signal. Since the pools D and U were not able to receive direct synaptic connections from the input neurons, the network also had to learn to communicate the presence of the cue pattern via disynaptic connections from the input neurons to these pools.

Network responses before and after learning are shown in \figref[B]{fig:exp-peters}. Initially, the rewarded goal was only reached occasionally, while the turnover of synaptic connections (number of synaptic connections that became functional or became non-functional in a time window of 2 hours) remained very high (see \figref[E]{fig:exp-peters}). After about 3~h, performance improved drastically (\figref[C]{fig:exp-peters}), and simultaneously the turnover of synaptic connections slowed down (\figref[E]{fig:exp-peters}). 
After learning for 8 hours, the network was able to solve the task in most of the trials, and the average trial duration (movement completion time) had decreased to less than 1 second ($851\pm 46~\text{ms}$, \figref[C]{fig:exp-peters}). Improved performance was accompanied by more stereotyped network activity and lever movement patterns as in the experimental data of \cite{PetersETAL:14}: compare our \figref[D]{fig:exp-peters} with Fig.~1b and Fig.~2j of \cite{PetersETAL:14}. In \figref[D]{fig:exp-peters} we show the trial-averaged activity of the 60 excitatory neurons before and after learning for 22 hours. The neurons are sorted in the first two plots of \figref[D]{fig:exp-peters} by the time of maximum activity after movement onset times before learning, and in the 3rd plot resorted according to times of maximum activity after 22 hours of learning
(see {\em Methods}). These plots show that reward-based learning led to a restructuring of the network activity: an assembly of neurons emerged that controlled a sharp upwards movement. Also, less background activity was observed after 22 hours of learning, in particular for neurons with early activity peaks. Lower panels in \figref[D]{fig:exp-peters} show the average lever movement and 10 individual movement traces at the beginning and after 22 hours of learning. Similar as in \cite{PetersETAL:14} the lever movements became more stereotyped during learning, featuring a sharp upwards movement at cue onset followed by a slower downwards movement in preparation for the next trial.

\rokni{
The synaptic parameter drifts due to SDE \eqref{eq:sde} inherently lead to forgetting. In \figref[F]{fig:exp-peters} we tested this effect by running a very long experiment over 14 simulated days. After 24~h, when the network had learned to reliably solve the task, we stopped the application of the reward but continued the synaptic dynamics. We found that the task could be reliably recalled for more than 5 days. Onset of forgetting was observed after day 6. We wondered whether a simple consolidation mechanism could prevent forgetting in our model. To test this we used the prior distribution $\ps{\bth}$ to stabilize the synaptic parameters. After 4 simulated days we set the mean of the prior to the current value of the synaptic parameters and reduced the variance, while continuing the synaptic dynamics with the same temperature. A similar mechanism for synaptic consolidation has been recently suggested in \cite{KirkpatrickETAL:17}. This mechanism reliably prevents forgetting in our model throughout the simulation time of 14 days. We conclude that the effect of forgetting is quite mild in our model and can be further suppressed by a consolidation mechanism that stabilizes synapses on longer timescales.
}

Next we tested whether similar results could be achieved with a simplified version of the stochastic synapse dynamics while a potential synaptic connection $i$ is non-functional, i.e., $\theta_i\leq 0$. Eq.~\eqref{eq:sde} defines for such non-functional synapses an Ornstein-Uhlenbeck process, which yields a heavy-tailed distribution for the waiting time until reappearance \rokni{(\figref[G]{fig:exp-peters}, left)}. 
We tested whether similar learning performance can be achieved if one approximates the distribution by an exponential distribution, for which we chose a mean of 12~h. The small distance between the blue and green curve in \figref[C]{fig:exp-peters} shows that this is in fact the case for the overall computational task that includes a task switch at 24~h that we describe below. \david{Compensation for the task switch was slightly slower when the approximating exponential distribution was used, but the task performance converged to the same result as for the exact rule.}  This holds in spite of the fact that the approximating exponential distribution is less heavy-tailed (\figref[G]{fig:exp-peters}, right)}. Altogether these results show that rewiring and synaptic plasticity according to Eq.~\eqref{eq:sde} yields self-organization of a generic recurrent network of spiking neurons so that it can control an arbitrarily chosen motor control task. 

\subsection*{Compensation for network perturbations}

We wondered whether this model for the task of [Peters et al., 2014] would in addition be able to compensate for a drastic change in the task, an extra challenge that had not been considered in the experiments of \cite{PetersETAL:14}. To test this we suddenly interchanged after 24~h the actions that were triggered by the pools D and U. D now caused upwards and U downwards lever movement.

We found that our model compensated immediately (see the faster movement in the parameter space depicted in \figref[H]{fig:exp-peters}) for this perturbation and reached after about 8~h a similar performance level as before (\figref[C]{fig:exp-peters}). This compensation phase was accompanied by a substantial increase in the turnover of synaptic connections similar as in experiments for learning of a new task, see e.g. \cite{XuETAL:09} (\figref[E]{fig:exp-peters}). The turnover rate also remained slightly elevated during the subsequent learning period. Furthermore, a new assembly of neurons emerged that now triggered a sharp onset of activity in the pool {D} (compare the activity neural traces at h = 22 and h = 46 in \figref[D]{fig:exp-peters}). \david{Another experimentally observed phenomenon that occured in our model were drifts of neural codes, which happened also during phases of the experiment without perturbations. Despite these drifts,} the task performance stayed constant, similar to experimental data in \cite{DriscollHarvey:16} (see \figref{fig:peths} in \nameref{sec:methods}).

If rewiring was disabled after the task change at 24~h the compensation was significantly delayed and overall performance declined (see black curve in \figref[C]{fig:exp-peters}). Here, we disallowed any turnover of potential synaptic connections such that the connectivity remained the same after 24~h. This result suggests that rewiring is necessary for adapting to the task change. In \rokni{\figref[H]{fig:exp-peters}} we further analyzed the profile of synaptic turnover for the different populations of the network scaffold in \figref[A]{fig:exp-peters}. The synaptic parameters were measured immediately before the task change at 24~h and compared to the connectivity after compensation at 48~h for the experiment shown in \figref[C]{fig:exp-peters} (blue). Most synapses (66-75\%) were non-functional before and after the task change (stable non-functional). About 20\% of the synapses changed their behavior and either became functional or non-functional. Most prominently a large fraction (21.9\%) of the synapses from hidden neurons to U became non-functional while only few (5.9\%) new connections were introduced. The connections from hidden to D showed the opposite behavior. This modification of the network connectome reflects the requirement to reliably route information about the presence of the cue pattern encoded in the activity of hidden neurons to the pool D (and not to U) to initiate the lever movement after the task change.

\legi{We then asked whether rewiring is also necessary for the initial learning of the task. To answer this question, we performed a simulation where the network connectivity was fixed from the beginning. We found that inital task performance was not significantly worse compared to the setup with rewiring. This indicates that at least for this task, rewiring is necessary for compensating task switches, but not for initial task learning. We expect however that this is not the case for more complex tasks, as indicated by a recent study that used artificial neural networks \cite{bellec2017deep}. }

A structural difference between stochastic learning models such as \legi{policy sampling} and learning models that focus on convergence of parameters to a (locally) optimal setting becomes apparent when one tracks the temporal evolution of the network parameters $\bth$ over larger periods of time during the previously discussed learning process \rokni{(\figref[I]{fig:exp-peters})}. Although performance no longer improved after 5~hours, both network connectivity and parameters kept changing in task-irrelevant dimensions.
For \rokni{\figref[I]{fig:exp-peters}} we randomly selected 5$\%$ of the roughly 47000 parameters $\theta_i$ and plotted the first 3 principal components of their dynamics. The task change after 24 hours caused the parameter vector $\bth$ to migrate to a new region within about 8 hours of continuing learning \david{(see \figref{fig:pca-proj} where the projected parameter dynamics is further analyzed)}. Again we observe that \legi{policy sampling} keeps exploring different equally good solutions after the learning process has reached stable performance.

\david{
To further investigate the impact of the temperature parameter $T$ on the exploration speed in the parameter space we measured the amplitudes of parameter changes for different values of $T$. We recorded the synaptic parameters every 20 minutes and estimated the exploration speed by measuring the average Euclidean distance between successive snapshots of the parameter vectors. We found that a temperature of $T=0.1$ resulted in an increase in exploration speed by around 150\% compared to the case of $T=0.0$. A temperature of $T=0.5$ resulted in an increase in exploration speed by around 400\%.
}

\subsection*{Relative contributions of spontaneous and activity-dependent synaptic processes}

\begin{figure}[hp]
\begin{center}
  \includegraphics[width=\textwidth]{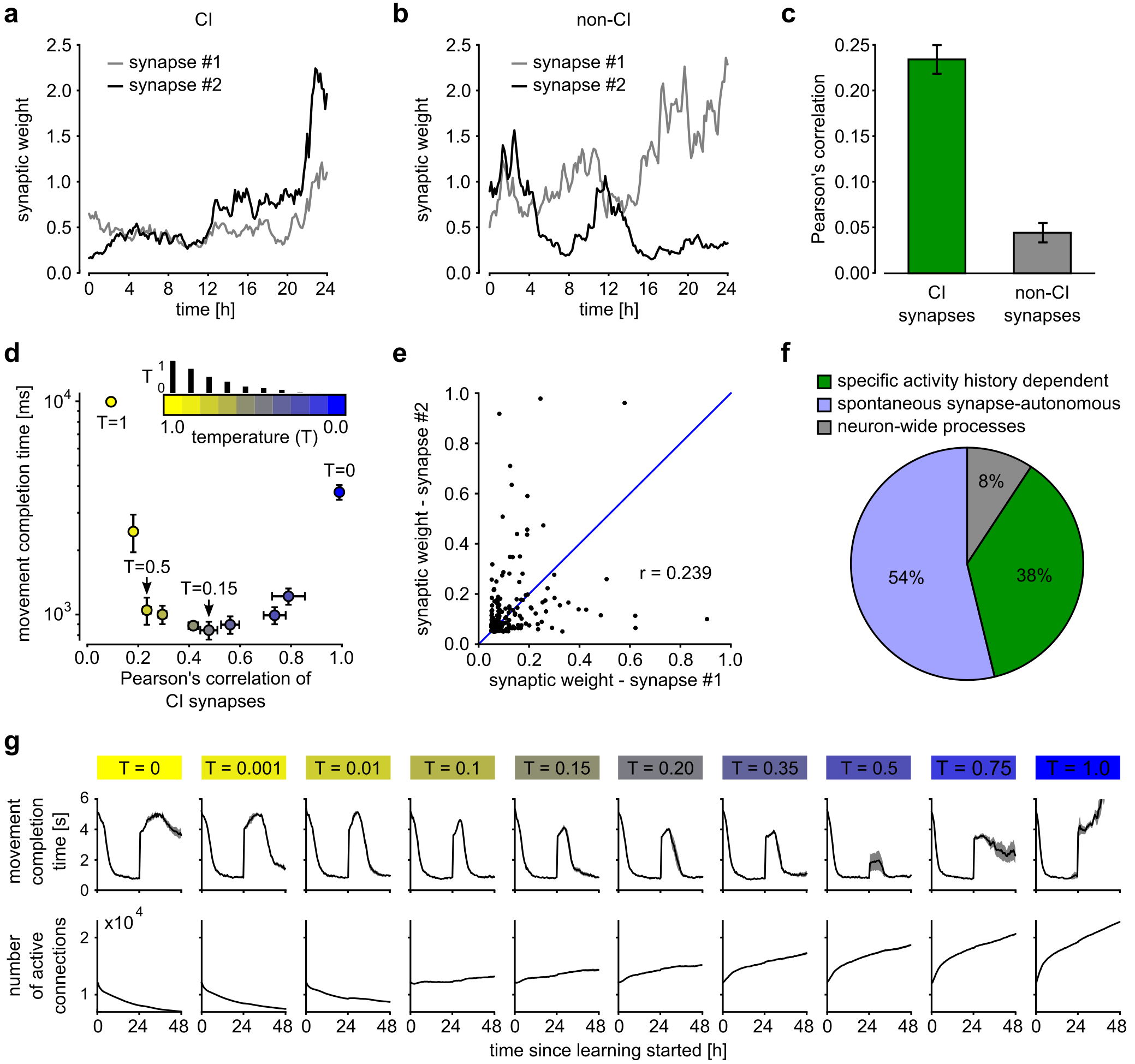}
\end{center}
\caption{{\bf Contribution of spontaneous and neural activity-dependent processes to synaptic dynamics}
\pl{A,B}
Evolution of synaptic weights $w_i$ plotted against time for a pair of CI synapses in \pr{A}, and non-CI synapses in \pr{B}, for temperature $T=0.5$.
\pl{C}
Pearson's correlation coefficient computed between synaptic weights of CI and non-CI synapses of a network with $T=0.5$ after 48 h of learning as in \figref[C,D]{fig:exp-peters}. CI synapses were only weakly, but significantly stronger correlated than non-CI synapses.
\pl{D}
Impact of $T$ on correlation of CI synapses (X-axis) and learning performance (Y-axis). Each dot represents averaged data for one particular temperature value, indicated by the color. 
Values for $T$ were 1.0, 0.75, 0.5, 0.35, 0.2, 0.15, 0.1, 0.01, 0.001, 0.0.
\textit{Caption continued on next page...}
\label{fig:ci-vs-no-ci}
}
\end{figure}

\begin{figure}[ht]
\textit{Caption of Fig.~\ref{fig:ci-vs-no-ci} continued:}
These values are proportional to the small vertical bars above the color bar. The performance (measured in movement completion time) is measured after 48 hours for the learning experiment as in \figref[C,D]{fig:exp-peters}, where the network changed completely after 24~h. Good performance was achieved for a range of temperature values between 0.01 and 0.5. Too low $(<0.01)$ or too high $(>0.5)$ values impaired learning. Means + s.e.m. over 5 independent trials are shown.
\pl{E}
Synaptic weights of 100 pairs of CI synapses that emerged from a run with $T=0.5$. Pearson's correlation is 0.239, comparable to the experimental data in Fig.~8A-D of \cite{DvorkinZiv:16}.
\pl{F}
Estimated contributions of activity history dependent (green), spontaneous synapse-autonomous (blue) and neuron-wide (gray) processes to the synaptic dynamics for a run with $T=0.15$. The resulting fractions are very similar to those in the experimental data, see Fig.~8E of \cite{DvorkinZiv:16}.
\rokni{
\pl{G}
Evolution of learning performance and total number of active synaptic connections for different temperatures as in \pr{D}. Compensation for task perturbation was significantly faster with higher temperatures. Temperatures larger than 0.5 prevented compensation. Overall number of synapses was decreasing for temperatures $T < 0.1$ and increasing for $T \geq 0.1$.
}
\end{figure}

\cite{DvorkinZiv:16} analyzed the correlation of sizes of postsynaptic densities and spine volumes for synapses that shared the same pre- and postsynaptic neuron, called commonly innervated (CI) synapses, and also for synapses that shared in addition the same dendrite (\cisd{}). Activity-dependent rules for synaptic plasticity, such as Hebbian or STDP rules -- on which previous models for network plasticity relied -- suggest that the strength of CI and especially \cisd{} synapses should be highly correlated. But both data from ex-vivo \cite{KasthuriETAL:15} and neural circuits in culture \cite{DvorkinZiv:16} show that postsynaptic density sizes and spine volumes of \cisd{}  synapses are only weakly correlated, with correlation coefficients between 0.23 and 0.34. Thus even with a conservative estimate that corrects for possible influences of their experimental procedure, more than 50\% of the observed synaptic strength appears to result from activity-independent stochastic processes (Fig.~8E of \cite{DvorkinZiv:16}). \david{\cite{bartol2015nanoconnectomic} had previously found larger correlations of synaptic strengths of \cisd{} synapses for a smaller data set (based on 17 \cisd{} pairs instead of the 72 pairs, 10 triplets, and 2 quadruplets in the ex-vivo data from \cite{KasthuriETAL:15}). But the spine volumes differed in these pairs also on average by a factor of around 2.} 

We asked how such a strong contribution of activity-independent synaptic dynamics affects network learning capabilities, such as the ones that were examined in Fig.~\ref{fig:exp-peters}. We were able to carry out this test because many synaptic connections between neurons that were formed in our model consisted of more than one synapse. We classified pairs of synapses that had the same pre- and postsynaptic neuron as CI synapses (one could also call them \cisd{} synapses, since the neuron model did not have different dendrites), and pairs with the same postsynaptic but different presynaptic neurons as non-CI synapses. Example traces of synaptic weights for CI and non-CI synapse pairs of our network model from Fig.~\ref{fig:exp-peters} are shown in \figref[A,B]{fig:ci-vs-no-ci}. CI pairs were found to be more strongly correlated than non-CI pairs (\figref[C]{fig:ci-vs-no-ci}). However also the correlation of CI pairs was quite low, and varied with the temperature parameter $T$ in Eq.~(\ref{eq:sde}), see  \figref[D]{fig:ci-vs-no-ci}. The correlation was measured in terms of the  Pearson correlation (covariance of synapse pairs normalized between -1 and 1).

Since the correlation of CI pairs in our model depends on the temperate $T$, we analyzed the model of Fig.~\ref{fig:exp-peters} for different temperatures (the temperature had been fixed at T=0.1 throughout the experiments for Fig.~\ref{fig:exp-peters}).
In \figref[D]{fig:ci-vs-no-ci} the Pearson correlation coefficient for CI synapses is plotted together with the average performance achieved on the task of \figref[D-H]{fig:exp-peters} \legi{($24$ hours after the task switch)} for networks with different temperatures $T$. The best performing temperature region for the task ($0.01 \leq T \leq 0.5$) roughly coincided with the region of experimentally measured values of Pearson's correlation for CI-synapses. \figref[E]{fig:ci-vs-no-ci} shows the correlation of 100 CI synapse pairs that emerged from a run with $T=0.5$. We found a value of $r=0.239$ in this case. This value is in the order of the lowest experimentally found correlation coefficients in \cite{DvorkinZiv:16} (both in culture and ex-vivo, see Fig.~8A-D in \cite{DvorkinZiv:16}). \rokni
{
The speed of compensation and the overall replenishing of synapses was strongly dependent on the temperature $T$ (see \figref[G]{fig:ci-vs-no-ci}). For $T=0$, a complete compensation for the task changes was prevented (performance converged to $2.5 \pm 0.2~s$ during a longer run of 96 hours). In other words, the temperature region $0.01 \leq T \leq 0.5$ -- that is consistent with experimentally measured Pearson's correlation for CI-synapses -- leads to fastest task re-learning, allowing for a compensation within approximately $12$ hours of exposure.}
For $T=0.15$ we found the best \legi{compensation capabilities} and the closest match to experimentally measured correlations when the results of \cite{DvorkinZiv:16} were corrected for measurement limitations: A correlation coefficient of $r=0.46 \pm 0.034$ for CI synapses and $0.08 \pm 0.015$ for non-CI synapse pairs (mean $\pm$ s.e.m. over 5 trials, \david{CI synapses were significantly stronger correlated than non-CI, p-value below 0.005 in all trials. Statistical significance values based on two-tailed Mann-Whitney U test.}).

\cite{DvorkinZiv:16} further analyzed the ratio of contributions from different processes to the measured synaptic dynamics. They analyzed the contribution of neural activity history dependent processes, which amount for 36\% of synapse dynamics in their data, and that of neuron-wide processes that were not specific to presynaptic activity, but specific to the activity of the postsynaptic neuron (8\%). Spontaneous synapse-autonomous processes were found to explain 56\% of the observed dynamics (see Fig.~8E in \cite{DvorkinZiv:16}). The results from our model, that are plotted in \figref[F]{fig:ci-vs-no-ci}, match these experimentally found values quite well. Altogether we found that the results of \cite{DvorkinZiv:16} are best explained by our model for a temperature parameter between $T=0.5$ (corresponding to their lowest measured correlation coefficient) and $T=0.15$ (corresponding to their most conservative estimate). This range of parameters coincided with well-functioning learning behavior in our model, which included a test of compensation capability for a change of the task after 24~h  (\figref[D]{fig:ci-vs-no-ci}). Hence our model suggests that a large component of stochastic synapse-autonomous processes, as it occurs in the data, supports efficient network learning and compensation for changes in the task.

\section*{\bf Discussion}

Recent experimental data (\cite{DvorkinZiv:16}, where in Fig.~8 also mouse brain data from \cite{KasthuriETAL:15} were reanalyzed, \david{and \cite{NagaokaETAL:16}}) suggest that common models for learning in neural networks of the brain need to be revised, since synapses are subject to powerful processes that do not depend on pre- and postsynaptic neural activity. 
In addition, experimentally found network rewiring has so far not been integrated into models for reward-gated network plasticity. We have presented a theoretical framework that enables us to investigate and understand reward-based network rewiring and synaptic plasticity in the context of the experimentally found high level of activity-independent fluctuations of synaptic connectivity and synaptic strength. We have shown that the analysis of the stationary distribution of network configurations, in particular the Fokker-Planck equation from theoretical physics, allows us to understand how large numbers of local stochastic processes at different synapses can orchestrate global goal-directed network learning. This approach provides a new normative model for reward-gated network plasticity.  

We have shown in \figref{fig:pattern-classification} that the resulting model is consistent with experimental data on dopamine-dependent spine dynamics reported in \cite{YagishitaETAL:14}, and that it provides an understanding how these local stochastic processes can produce function-oriented cortical-striatal connectivity. We have shown in  \figref{fig:exp-peters} that this model also elucidates reward-based self-organization of generic recurrent neural networks for a given computational task. We chose as benchmark task the production of a specific motor output in response to a cue, like in the experiments of \cite{PetersETAL:14}. Similarly as reported in \cite{PetersETAL:14}, the network connectivity and dynamics reorganized itself in our model, just driven by stochastic processes and rewards for successful task completion, and reached a high level of \david{performance}. Furthermore it maintained this computational function in spite of continuously ongoing further rewiring and network plasticity.  A quantitative analysis of the impact of stochasticity on this process has shown in \figref[]{fig:ci-vs-no-ci} that the network learns best when the component of synaptic plasticity that does not depend on neural activity is fairly large, as large as reported in the experimental data of \cite{KasthuriETAL:15,DvorkinZiv:16}. 

Our approach is based on experimental data for the biological implementation level of network plasticity, i.e., for the lowest level  of the Marr hierarchy of models \cite{MarrPoggio:76}. However, we have shown that these experimental data have significant implications for understanding network plasticity on the top level ("what is the functional goal?") and the intermediate algorithmic level ("what is the underlying algorithm?") of the Marr hierarchy. They suggest for the top level that the goal of network plasticity is to sample from a posterior distribution of network configurations.
\legi{This posterior integrates functional demands formalized by the expected discounted reward $\mathcal{V}(\bth)$ with a prior $p_S(\bth)$ in a multiplicative manner $p^*(\bth) \;\propto\; p_S(\bth) \, \times \, \mathcal{V}(\bth)$.} 
Priors can represent structural constraints as well as results of preceding learning experiences and innate programs. \legi{Since our model samples from a distribution proportional to $p^*(\bth)^{\frac{1}{T}}$, for $T=1$, our model suggests to view reward-gated network plasticity as Bayesian inference over network configurations on a slow time scale, see {\em Probabilistic framework for reward-modulated learning} in {\em Methods} for details. For a temperature parameter $T\neq 1$, the model samples from a tempered version of the posterior, which generalizes the basic Bayesian approach.} This Bayesian perspective also creates a link to previous work on Bayesian reinforcement learning \cite{VlassisETAL:12,RawlikETAL:13}. \legi{We note however that we do not consider parameter adaptation in our framework to implement full Bayesian learning, as there is no integration over the posterior paramter settings to obtain network outputs (or actions in a reinforcement learning context). Even if one would do that, it would be of little practical use, since the sampling would be much too slow in any but the simplest networks.}
The experimental data suggest for the intermediate algorithmic level of the Marr hierarchy a strong reliance on stochastic search ("synaptic sampling"). The essence of the resulting model for reward-gated network learning is illustrated in \figref{fig:model-illustration}: The traditional view of deterministic gradient ascent (policy gradient) in the landscape (panel {b}) of reward expectation is first modified through the integration of a prior (panel {c}), and then through the replacement of gradient ascent by continuously ongoing stochastic sampling (\legi{policy sampling}) from the posterior distribution of panel {d}, which is illustrated in panels {e} and {f}. 

\david{This model explains a number of experimental data that had not been addressed by previous models.} Continuously ongoing stochastic sampling of network configurations suggests that synaptic connectivity does not converge to a fixed point solution but rather undergoes permanent modifications (\figref[H,I]{fig:exp-peters}). This \legi{behavior} is compatible with reports of continuously ongoing spine dynamics and axonal sprouting even in the adult brain \cite{HoltmaatSvoboda:09, YasumatsuETAL:08, StettlerETAL:06, YamahachiETAL:09, LoewensteinETAL:11, HoltmaatETAL:05, LoewensteinETAL:15}. \david{Recently proposed models to maintain stable network function in the presence of highly volatile spine dynamics suggest that subsets of connections are selectively stabilized to support network function \cite{BerryNedivi:17,MongilloETAL:17}. Our result shows that high task performance can be reached in spiking neural networks in the presence of high volatility of all synapses. Still our model can be extended with a process that selectively stabilizes synapses on longer timescales as demonstrated in \figref[F]{fig:exp-peters}. In addition, our model predicts that not only synaptic spine dynamics but also \emph{changes} of synaptic efficacies show a large stochastic component on all timescales.}

\david{The continuously ongoing parameter changes induce continuously ongoing changes in the assembly sequences that accompany and control a motor response (see \figref[D]{fig:exp-peters}). These changes do not impair the performance of the network, but rather enable the network to explore different but equally good solutions when exposed for many hours to the same task (see \figref[I]{fig:exp-peters}).} Such continuously ongoing drifts of neural codes in functionally less relevant dimensions have already been observed experimentally in some brain areas \cite{ZivETAL:13,DriscollHarvey:16}. 
\david{Our model also suggests that the same computational function is realized by the same neural circuit in different individuals with drastically different parameters, a feature which has already been addressed in \cite{TangETAL:10, GrashowETAL:10, Marder:11, PrinzETAL:04}.  
In fact, this \emph{degeneracy} of neural circuits is thought to be an important property of biological neural networks \cite{Marder:11, PrinzETAL:04, MarderGoaillard:06}. Our model networks automatically compensate for disturbances by moving their continuously ongoing sampling of network configurations to a new region of the parameter space, as illustrated by the response to the disturbance marked by $\ast$ in \figref[I]{fig:exp-peters}.}

\rokni{Our theoretical framework is consistent with experimental data that showed drifts of neural representations in motor learning \cite{RokniETAL:07}. In that article, a stochastic plasticity model was proposed that is structurally similar to our model. It was shown in computer simulations that a simple feed forward rate-based neural network is able to retain stable functionality despite of such stochastic parameter changes. The authors hypothesized that this is the case because network parameters move on a sub-manifold in parameter space with constant performance. Our theoretical framework provides a mathematical justification for their hypothesis in general, but also refines these statements. It shows that the network samples network configurations (including the rewiring of connections that was not considered in \cite{RokniETAL:07}) from a well-defined distribution. The manifold that is visited during the learning process is given by the high-probability regions of this distribution, but in principle, also sub-optimal regions could be visited. Such sub-optimal regions are however highly unlikely if the parameter space is overcomplete, i.e., if large volumes of the parameter space lead to good performance. Hence, in comparison with \cite{RokniETAL:07}, this work provides the following features: (a) It provides a quantitative mathematical framework for the qualitative descriptions in \cite{RokniETAL:07} that allows a rigorous understanding of the plasticity processes, (b) it includes synaptic rewiring, reproducing experimental data on this topic and providing a hypothesis on its computational role, and (c), it is able to tackle the case of recurrent spiking neural networks as compared to feed forward rate models.}

\david{We have shown in \figref[F]{fig:exp-peters} that despite these permanent parameter drifts, the task performance in our model remains stable for many simulated days if reward delivery is stopped. At the same time, the model is also able to continuously adapt to  changes in the task \figref[c-e]{fig:exp-peters}. These results suggest that our model keeps a quite good balance between stability and plasticity \cite{AbrahamRobins:05}, which has already been shown to be one important functional aspect of network rewiring \cite{FauthETAL:15}. Furthermore, we have shown in \figref[F]{fig:exp-peters}, that the structural priors over synaptic parameters can be utilized to stabilize synaptic parameters similar to previous models of synaptic consolidation \cite{FusiETAL:05,KirkpatrickETAL:17}. In addition, more complex prior distributions over multiple synapses could be utilized to model homeostatic processes and clustering of synapses. The latter has been suggested as a mechanism to tackle the stability-plasticity dilemma \cite{FaresStepanyants:09}.}

In conclusion the mathematical framework presented in this article provides a principled way of understanding the complex interplay of deterministic and stochastic processes that underlie the implementation of goal-directed learning in neural circuits of the brain. It also offers a solution to the problem how reliable network computations can be achieved with a ``dynamic connectome'' \cite{RumpelTriesch:16}. We have argued that the stationary distribution of the high-dimensional parameter vector $\bth$ that results from large numbers of local stochastic processes at the synapses provides a time-invariant perspective of salient properties of a network. Standard reward-gated plasticity rules can achieve that this stationary distribution has most of its mass on regions in the parameter space that provide good network performance. The stochastic component of synaptic dynamics can flatten or sharpen the resulting stationary distribution, depending on whether the scaling parameter $T$ (``temperature'') of the stochastic component is larger or smaller than 1. A functional benefit of this stochastic component is that the network keeps exploring its parameter space even after a well-performing region has been found\legi{, providing one mechanism to tackle the exploration-exploitation dilemma (see \figref[J]{fig:pattern-classification})}. This enables the network to migrate quickly and automatically to a better performing region when the network or task \david{changes}. We found in the case of the motor learning task of Fig.~\ref{fig:exp-peters} that a temperature $T$ around 0.15, which lies in the same range as related experimental data (see Fig.~\ref{fig:ci-vs-no-ci}d), suffices to provide this functionally important compensation capability. \legi{The same mathematical framework can also be applied to artificial neural networks, leading to a novel brain-inspired learning algorithm that uses rewiring to train deep networks under the constraint of very sparse connectivity \cite{bellec2017deep}.}
(1657 words)

\section*{Methods}
\label{sec:methods}

\subsection*{\legi{Probabilistic} framework for reward-modulated learning}
\label{sec:correctness-synaptic-plasticity}
The classical goal of reinforcement learning is to maximize the expected future discounted reward $\mathcal{V}(\bth)$ given by 
\begin{equation}
  \mathcal{V}(\bth) \;=\; \expect[p(\rseq | \bth)]{ \int_{0}^\infty e^{-\frac{\tau}{\tau_e}} \,\rtau \; \d \tau } \;.
  \label{eqn:reward-prob-factorized}
\end{equation}
In Eq.~\eqref{eqn:reward-prob-factorized} we integrate over all future rewards $\rtau$, while discounting more remote rewards exponentially with a discount rate $\tau_e$, which for simplicity was set equal to $1~\unit{s}$ in this paper. We find (see Eq.~\eqref{eqn:eligibility-trace}) that this time constant $\tau_e$ is immediately related to the experimentally studied time window or eligibility trace for the influence of dopamine on synaptic plasticity \cite{YagishitaETAL:14}. \david{This property is true in general for reward-based learning rules that make use of eligibility traces and is not unique to our model.}
The expectation in Eq.~\eqref{eqn:reward-prob-factorized} is taken with respect to the distribution $p(\rseq | \bth)$ over sequences $\rseq = \{ \rtau, \tau \geq 0 \}$ of future rewards that result from the given set of synaptic parameters $\bth$. The stochasticity of the reward sequence $\rseq$ arises from stochastic network inputs, stochastic network responses, and stochastic reward delivery. The resulting distribution $p(\rseq | \bth)$ of reward sequences $\rseq$ for the given parameters $\bth$ can also include influences of network initial conditions by assuming some distribution over these initial conditions. Network initial conditions include for example initial values of neuron membrane voltages and refractory states of neurons. The role of initial conditions on network learning is further discussed below when we consider the online learning scenario in {\em Reward-modulated synaptic plasticity approximates gradient ascent on the expected discounted reward}.

There exists a close relationship between reinforcement learning and Bayesian inference \cite{VlassisETAL:12,RawlikETAL:13,BotvinickToussaint:12}. To make this relationship apparent, we define our model for reward-gated network plasticity by introducing a binary random variable $\rbin$ that represents the currently expected future discounted reward in a probabilistic manner. The likelihood $\pn {\rbin=1}{\bth}$ is determined in this theoretical framework by the expected future discounted reward Eq.~\eqref{eqn:reward-prob-factorized} that is achieved by a network with parameter set $\bth$ (see e.g., \cite{RawlikETAL:13}):
\begin{equation}
  \pn{\rbin=1}{\bth} \;\equiv\; \frac{1}{\mathcal{Z}_{\mathcal{V}}} \mathcal{V}(\bth) \;,
  \label{eqn:meth-reward-prob-factorized2}
\end{equation}
where $\mathcal{Z}_{\mathcal{V}}$ denotes a constant, that assures that Eq.~\eqref{eqn:meth-reward-prob-factorized2} is a correctly normalized probability distribution.
Thus reward-based network optimization can be formalized as maximizing the likelihood $\pn {\rbin=1}{\bth}$ with respect to the network configuration $\bth$. Structural constraints can be integrated into a stochastic model for network plasticity through a prior $p_S(\bth)$ over network configurations. Hence reward-gated network optimization amounts from a theoretical perspective to learning of  the posterior distribution $p^*(\bth | \rbin=1)$, which by Bayes' rule is defined (up to normalization) by $\ps{\bth} \cdot \pn{\rbin=1}{\bth}$. Therefore, the learning goal can be formalized in a compact form as evaluating the posterior distribution $p^*(\bth | \rbin=1)$ of network parameters $\bth$ under the constraint that the abstract learning goal $\rbin=1$ is achieved.

More generally, one is often interested in a tempered version of the posterior 
\begin{equation}
\label{eq:annealed_post}
p_T^*(\bth) \equiv \frac{1}{\mathcal{Z}} p^*(\bth | \rbin=1)^{\frac{1}{T}} \; ,
\end{equation}
where $\mathcal{Z}$ is a suitable normalization constant and $T>0$ is the temperature parameter that controls the ``sharpness'' of $p_T^*(\bth)$. For $T=1$, $p_T^*(\bth)$ is given by the original posterior, $T<1$ emphasizes parameter values with high probability in the posterior, while $T>1$ leads to parameter distributions $p_T^*(\bth)$ which are more uniformly distributed than the posterior.   

\subsection*{Analysis of \legi{policy sampling}}
\label{sec:theorem-1}
Here we prove that the stochastic parameter dynamics Eq.~\eqref{eq:sde} samples from the tempered posterior distribution $p_T^*(\bth)$ given in Eq.~\eqref{eq:annealed_post}. In {\em Results} we suppressed time-dependencies in order to simplify notation. We reiterate Eq.~\eqref{eq:sde-reduced} with explicit time-dependencies of parameters:
\begin{equation}
 d \thi \;=\; \beta \, \left . \ddthetai \log p^*(\bth|\rbin=1) \right|_{\btht}   dt  \;+ \; \sqrt{2 \beta T} \, d \wiener_{i} \;,
 \label{eq:sde_special}
\end{equation}

\noindent where the notation $\left . \ddthetai \,f(\bth) \right|_{\btht}$ denotes the derivative of $f(\bth)$ with respect to $\theta_i$ evaluated at the current parameter values $\btht$. By Bayes' rule, the derivative of the log posterior is the sum of the derivatives of the prior and the likelihood:
\begin{equation*}
  \ddthetai \log p^*(\bth|\rbin=1) \;=\; \ddthetai \log \ps{\bth} + \ddthetai \log \pn{\rbin=1}{\bth} \;=\; \ddthetai \log \ps{\bth} + \ddthetai \log \mathcal{V}(\bth)\;, 
\end{equation*}

\noindent which allows us to rewrite Eq.~\eqref{eq:sde_special} as
\begin{equation}
d \thi \;=\; \beta \, \left( \left . \ddthetai \,\log \ps{\bth} \right|_{\btht} \, + \, \left . \ddthetai \log \mathcal{V}(\bth)\right|_{\btht}  \right)  dt  \;+ \; \sqrt{2 \beta T} \, d \wiener_{i} \;,
\label{eq:sde-meth}
\end{equation}
which is identical to the form Eq.~\eqref{eq:sde}, where the contributions of $\ps{\bth}$ and $\mathcal{V}(\bth)$ are given explicitly.

 The fundamental property of the synaptic sampling dynamics  Eq.~\eqref{eq:sde_special} is formalized in Theorem \ref{lem:single_sup} and proven below. Before we state the theorem, we briefly discuss its statement in simple terms. Consider some initial parameter setting $\bth(0)$. Over time, the parameters change according to the dynamics \eqref{eq:sde_special}. Since the dynamics include a noise term, the exact value of the parameters $\bth(t)$ at some time $t>0$ cannot be determined. However, it is possible to describe the exact distribution of parameters for each time $t$. 
We denote this distribution by $p_\text{FP}(\bth, t)$, where the ``FP'' subscript stands for ``Fokker-Planck'' since the evolution of this distribution is described by the Fokker-Planck equation \eqref{eq:sup_FP} given below. Note that we make the dependence of this distribution on time explicit in this notation. 
It can be shown that for the dynamics \eqref{eq:sup_FP}, $p_\text{FP}(\bth, t)$ converges to a well-defined and unique {\em stationary distribution} in the limit of large $t$. 
Of practical relevance is the so-called burn-in time after which the distribution of parameters is very close to the stationary distribution.
Note that the parameters will continue to change. Nevertheless, at any time $t$ after the burn in, we can expect the parameter vector $\bth(t)$ to be situated at a particular value with the probability (density) given by the stationary distribution, see \figref[D,F]{fig:model-illustration}. Any distribution that is {\em invariant} under the parameter dynamics is a stationary distribution. Here, invariance means: when one starts with an invariant distribution over parameters in the Fokker-Planck equation, the dynamics are such that this distribution will be kept forever (we will use this below in the proof of Theorem \ref{lem:single_sup}). Theorem \ref{lem:single_sup} states that the parameter dynamics leaves $p_T^*(\bth)$ given in Eq.~\eqref{eq:annealed_post} invariant, i.e., it is a stationary distribution of the network parameters. Note that in general, the stationary distribution may not be uniquely defined. That is, it could happen that for two different initial parameter values, the network reaches two different stationary distributions. Theorem \ref{lem:single_sup} further states that for the synaptic sampling dynamics, the stationary distribution is unique, i.e., the distribution $p_T^*(\bth)$ is reached from {\em any} initial parameter setting when the conditions of the theorem apply.
We now state Theorem \ref{lem:single_sup} formally. 
To simplify notation we drop in the following the explicit time dependence of the synaptic parameters $\bth$.
\begin{thm}
\label{lem:single_sup}
Let $p^{*} (\bth \,|\, \rbin=1)$ be a strictly positive, continuous probability distribution over parameters $\bth$, twice continuously differentiable with respect to $\bth$, and let $\beta>0$. Then
the set of stochastic differential equations Eq.~\eqref{eq:sde_special}
leaves the distribution $p_T^{*}(\bth)$ \eqref{eq:annealed_post} invariant.
Furthermore, $p_T^{*}(\bth)$ is the unique stationary distribution of the sampling dynamics.
\end{thm}

\begin{proof}
The proof is analogous to the one provided in \cite{KappelETAL:15}. The stochastic differential equation Eq.~\eqref{eq:sde_special} translates into a Fokker-Planck equation \cite{Gardiner:04} that describes the evolution of the distribution over parameters $\bth$
\begin{equation}\label{eq:sup_FP}
 \ddt p_{\text{FP}}(\bth, t) = \sum_{i} - \ddthetai  \left(\beta \, \ddthetai \log p^{*} (\bth \,|\, \rbin=1) \right) p_{\text{FP}}(\bth, t)
 + \ddthetaisq \left(\beta \, T \,  p_{\text{FP}}(\bth, t)\right),
\end{equation}
where $p_{\text{FP}}(\bth, t)$ denotes the distribution over network parameters at time $t$. To show that $p_T^*(\bth)$ leaves the distribution invariant, we have to show that $\ddt p_{\text{FP}}(\bth,t)=0$ (i.e., $p_{\text{FP}}(\bth,t)$ does not change) if we set $p_{\text{FP}}(\bth,t)$ to $p_T^*(\bth)$ on the right hand side of Eq.~\eqref{eq:sup_FP}.
Plugging in the presumed stationary distribution $p_T^{*}(\bth)$ for $p_{\text{FP}}(\bth,t)$ on the right hand side of Eq.~\eqref{eq:sup_FP}, one obtains
\begin{align*}
 \ddt p_{\text{FP}}(\bth,t) &= \sum_{i} - \ddthetai \left(\beta \, \ddthetai \log p^{*} (\bth \,|\, \rbin=1) \, p_T^{*}(\bth)  \right) 
 + \ddthetaisq \left(\beta \, T \, p_T^{*}(\bth) \right)\\
&=  \sum_{i} - \ddthetai \left(\beta \, p_T^{*}(\bth) \, \ddthetai \log p^{*} (\bth \,|\, \rbin=1)  \right)  
 + \ddthetai \left(\beta \, T \, \ddthetai p_T^{*}(\bth)  \right)\\
&= \sum_{i} - \ddthetai \left(\beta  \, p_T^{*}(\bth)  \, \ddthetai \log p^{*} (\bth \,|\, \rbin=1) \right)  
 + \ddthetai \left(\beta \, T \, p_T^{*}(\bth) \, \ddthetai \log p_T^{*}(\bth)  \right)\;,
 \end{align*} which by inserting
$p_T^{*}(\bth)  \;=\; \frac{1}{\mathcal{Z}} p^{*} (\bth \,|\, \rbin=1)^\frac{1}{T}$, with normalizing constant $\mathcal{Z}$,
becomes
 \begin{align*}
\ddt p_{\text{FP}}(\bth,t) &= \frac{1}{\mathcal{Z}} \sum_i - \ddthetai \left(\beta  \, p^{*}(\bth) \, \ddthetai \log p^{*} (\bth \,|\, \rbin=1) \right)  
 + \ddthetai \left(\beta \, T \, p^{*}(\bth) \, \frac{1}{T} \ddthetai \log p^{*} (\bth \,|\, \rbin=1)  \right)\\
&= \sum_i 0 = 0\;\;.
\end{align*}
This proves that $p_T^{*}(\bth)$ is a stationary distribution of the parameter sampling dynamics Eq.~\eqref{eq:sde_special}. Since $\beta$ is strictly positive, this stationary distribution is also unique (see Section 3.7.2 in \cite{Gardiner:04}).

The unique stationary distribution of Eq.~\eqref{eq:sup_FP} is given by $p_T^*(\bth) = \frac{1}{\mathcal{Z}} p^*(\bth | \rbin=1)^{\frac{1}{T}}$, i.e. $p_T^* (\bth)$ is the only solution for which $\ddt p_{\text{FP}}(\bth, t)$ becomes $0$, which completes the proof. 
\end{proof}

\david{
The uniqueness of the stationary distribution follows because each parameter setting can be reached from any other parameter setting with non-zero probability (ergodicity). The stochastic process can therefore not get trapped in cycles or absorbed into a subregion of the parameter space. The time spent in a certain region of the parameter space is therefore directly proportional to the probability of that parameter region under the posterior distribution. The proof requires that the posterior distribution is smooth and differentiable with respect to the synaptic parameters. This is not true in general for a spiking neural networks. In our simulations we used a stochastic neuron model (defined in the next section). As the reward landscape in our case is defined by the {\em expected} discounted reward (see below), a probabilistic network tends to smoothen this landscape and therefore the posterior distribution.
}

\subsection*{Network model}
\label{sec:meth-neuron-model}
Plasticity rules for this general framework were derived based on a specific spiking neural network model, which we describe in the following. All reported computer simulations were performed with this network model.
We considered a general network scaffold $\mathcal{N}$ of $K$ neurons 
with potentially asymmetric recurrent connections. Neurons are indexed in an arbitrary order by integers between $1$ and $K$. We denote the output spike train of a neuron $k$ by $z_k(t)$. It is defined as the sum of Dirac delta pulses positioned at the spike times $t_k^{(1)}, t_k^{(2)}, \dots$, i.e., $z_k(t) = \sum_l \delta(t-t_k^{(l)})$. Potential synaptic connections are also indexed in an arbitrary order by integers between $1$ and $K_\text{syn}$, where $K_\text{syn}$ denotes the number of potential synaptic connections in the network. We denote by $\preidef$ and $\postidef$ the index of the pre- and postsynaptic neuron of synapse $i$, respectively, which unambiguously specifies the connectivity in the network. Further, we define $\syn{k}$ to be the index set of synapses that project to neuron $k$. Note that this indexing scheme allows us to include  multiple (potential) synaptic connections between a given pair of neurons. We included this experimentally observed feature of biological neuronal networks in all our simulations. We denote by $w_{i}(t)$ the synaptic efficacy of the $i$-th synapse in the network at time $t$.

Network neurons were modeled by a standard stochastic variant of the spike response model  \cite{GerstnerETAL:14}. In this model, the membrane potential of a neuron $k$ at time $t$ is given by
\begin{equation}
u_k(t) \;=\; \sum_{i \,\in\, \syn{k}} \hz_{\prei}(t)\, w_{i}(t) \;+\; \vartheta_k(t) \; , \label{eq:membrane-potential}
\end{equation}
where $\vartheta_k(t)$ denotes the slowly changing bias potential of neuron $k$, and $\hz_{\prei}(t)$ denotes the trace of the (unweighted) postsynaptic potentials (PSPs) that neuron $\preidef$ leaves in its postsynaptic synapses at time $t$. More precisely, it is defined as $\hz_{\prei}(t) = z_{\prei}(t) \ast \epsilon(t)$ given by spike trains filtered with a PSP kernel of the form $\epsilon(t) = \Theta(t) \, \frac{\tau_r}{\tau_m - \tau_r} \left( e^{-\frac{t}{\tau_m}} - e^{-\frac{t}{\tau_r}}  \right)$, with time constants $\tau_m=20~\text{ms}$ and $\tau_r=2~\text{ms}$, if not stated otherwise. Here $\ast$ denotes convolution and $\Theta(\cdot)$ is the Heaviside step function, i.e. $\Theta(x)=1$ for $x\ge 0$ and $0$ otherwise.

The synaptic weights $w_{i}(t)$ in Eq.~\eqref{eq:membrane-potential} were determined by the synaptic parameters $\thi$ through the mapping Eq.~\eqref{eq:thetamap} for $\thi > 0$. Synaptic connections with $\thi \leq 0$ were interpreted as not functional (disconnected) and $w_{i}(t)$ was therefore set to 0 in that case.

The bias potential $\vartheta_k(t)$ in Eq.~\eqref{eq:membrane-potential} implements a slow adaptation mechanism of the intrinsic excitability, which ensures that the output rate of each neuron stays near the firing threshold and the neuron maintains responsiveness \cite{DesaiETAL:99,FanETAL:05}. We used a simple adaptation mechanism which was updated according to
\begin{equation}
  \tau_{\vartheta}\, \frac{\d \vartheta_k(t)}{\d t} \;=\; \nu_0 - z_k(t)  \;,
  \label{eqn:adaptation-current}
\end{equation}
where $\tau_{\vartheta} = 50~\text{s}$ is the time constant of the adaptation mechanism and $\nu_0 = 5~\text{Hz}$ is the desired output rate of the neuron. In our simulations, the bias potential  $\vartheta_k(t)$ was initialized at -3 and then followed the dynamics given in Eq.~\eqref{eqn:adaptation-current}. We found that this regularization significantly increased the performance and learning speed of our network model. \david{In \cite{RemmeWadman:12} a similar mechanism was proposed to balance activity in networks of excitatory and inhibitory neurons. The regularization used here can be seen as a simplified version of this mechanism that regulates the mean firing rate of each excitatory neuron using a simple linear control loop and thereby stabilizes the output behavior of the network.}

We used a simple refractory mechanism for our neuron model. The firing rate, or intensity, of neuron $k$ at time $t$ is defined by the function $f_k(t) \;=\; f(u_k(t), \rho_k(t))$, where $\rho_k(t)$ denotes a refractory variable that measures the time elapsed since the last spike of neuron $k$. We used an exponential dependence between membrane potential and firing rate, such that the instantaneous firing rate of the neuron $k$ at time $t$ can be written as
\begin{equation}
  f_k(t) \;=\; f(u_k, \rho_k) \;=\; \exp(u_k) \Theta(\rho_k-t_\text{ref})\;.
  \label{eq:activation-function}
\end{equation}
Furthermore, we denote by $f_{\posti}(t)$ the firing rate of the neuron postsynaptic to synapse $i$. If not stated otherwise we set the refractory time $t_\text{ref}$ to 5~ms. In addition, a subset of neurons was clamped to some given firing rates (input neurons), such that $f_k(t)$ of these input neurons was given by an arbitrary function. We denote the spike train from these neurons by $\ve{x}(t)$, the network input.

\subsection*{Synaptic dynamics for the reward-based synaptic sampling model}
\label{sec:spiking-synaptic-sampling}
Here, we provide additional details on how the synaptic parameter dynamics Eq.~\eqref{eq:sde} was computed. We will first provide an intuitive interpretation of the equations and then provide a detailed derivation in the next section.  The second term $\ddthetai \log \mathcal{V}(\bth)$ of Eq.~\eqref{eq:sde} denotes the gradient of the expected future discounted reward Eq.~\eqref{eqn:reward-prob-factorized}. In general, optimizing this function has to account for the case where rewards are provided after some delay period. It is well known that this \emph{distal reward problem} can be solved using plasticity mechanisms that make use of eligibility traces in the synapses that are triggered by near coincident spike patterns, but their consolidation into the synaptic weights is delayed and \david{gated} by the reward signal $r(t)$ \cite{SuttonBarto:98,Izhikevich:07}. The theoretically optimal shape for these eligibility traces can be derived using the reinforcement learning theory and depends on the choice of network model. For the spiking neural network model described above, the gradient $\ddthetai \log \mathcal{V}(\bth)$ can be estimated through a plasticity mechanism that uses an eligibility trace $e_{i}(t)$ in each synapse $i$ which gets updated according to
\begin{equation}
 \frac{\d e_{i}(t)}{\d t} \;=\; - \frac{1}{\tau_{e}} e_{i}(t) \;+\; w_{i}(t) \, \hz_{\prei}(t) \,  (\zkt -f_{\posti}(t))\;, \label{eqn:eligibility-trace}
\end{equation}
where $\tau_e=1~\text{s}$ is the time constant of the eligibility trace. Recall that $\preidef$ denotes the index of the presynaptic neuron and $\postidef$ the index of the postsynaptic neuron for synapse $i$. In Eq.~\eqref{eqn:eligibility-trace} $\zkt$ denotes the postsynaptic spike train, $\fkt$ denotes the instantaneous firing rate (Eq.~\eqref{eq:activation-function}) of the postsynaptic neuron and $w_{i}(t) \, \hz_{\prei}(t)$ denotes the postsynaptic potential under synapse $i$.

The last term of Eq.~\eqref{eqn:eligibility-trace} shares salient properties with standard STDP learning rules, since plasticity is enabled by the presynaptic term $\hz_{\prei}(t)$ and gated by the postsynaptic term $(\zkt -f_{\posti}(t))$ (see \cite{PfisterETAL:06}). The latter term also regularizes the plasticity mechanism such that synapses stop growing if the firing probability $f_{\posti}(t)$ of the postsynaptic neuron is already close to one.

The eligibility trace Eq.~\eqref{eqn:eligibility-trace} is \legi{multiplied} by the reward $r(t)$ and integrated in each synapse $i$ using a second dynamic variable
\begin{equation}
 \frac{\d g_{i}(t)}{\d t} \;=\; - \frac{1}{\tau_{g}} g_{i}(t) \;+\;\left( \frac{\rt}{\rhatt} + \alpha \right) \, e_{i}(t) \;,
 \label{eqn:gradient-est}
\end{equation}
\david{where $\rhatt$ is a low-pass filtered version of $\rt$ (described below).}
The variable $g_{i}(t)$ combines the eligibility trace and the reward, and averages over the time scale $\tau_g$. $\alpha$ is a constant offset on the reward signal. \david{This parameter can be set to an arbitrary value without changing the stationary dynamics of the model (see next section)}. In our simulations, this offset $\alpha$ was chosen slightly above $0$ ($\alpha=0.02$) such that small parameter changes were also present without any reward, as observed in \cite{YagishitaETAL:14}. \david{Furthermore, $\alpha$  does not have to be chosen constant. E.g. this term can be used to incorporate predictions about the reward outcome by setting $\alpha$ to the negative of output of a \emph{critic} network that learns to predict future reward. This approach has been previously studied in \cite{FremauxETAL:13} to model experimental data of \cite{SchultzETAL:97,Schultz:02}.}

\david{In the next section we show that $g_{i}(t)$ approximates the gradient of the expected future reward with respect to the synaptic parameter.} In our simulations we found that incorporating the low-pass filtered eligibility traces (Eq.~\eqref{eqn:gradient-est}) into the synaptic parameters works significantly better than using the eligibility traces directly for weight updates, although the latter approach was taken in a number of previous studies (see e.g. \cite{PfisterETAL:06,LegensteinETAL:08,UrbanczikSenn:09}). \david{Eq.~\eqref{eqn:gradient-est} essentially combines the eligibility trace with the reward and smoothens the resulting trace with a low-pass filter with time constant $\tau_g$. This time constant has been chosen to be in the order spontaneous decay of disinhibited CaMKII in the synapse which is closely related to spine enlargement in the dopamine-gated STDP protocol of  \cite{YagishitaETAL:14} (c.f. their Fig.~3F and Fig.~4C).}

$\rhatt$ in Eq.~\eqref{eqn:gradient-est} is a low-pass filtered version of $\rt$ that scales the synaptic updates. It was implemented through $\tau_a \frac{\d \rhatt}{\d t} = -\rhatt + \rt$, with $\tau_a=50~\text{s}$. \david{The value of $\tau_a$ has been chosen to equal $\tau_g$ based on theoretical considerations (see \emph{Online learning}).} This scaling of the reward signal has the following effect. If the current reward  $\rt$ exceeds the average reward $\rhatt$, the effect of the neuromodulatory signal $\rt$ will be greater than $1$. On the other hand, if the current reward is below average synaptic updates will be weighted by a term significantly lower than $1$. Therefore, parameter updates are preferred for which the current reward signal exceeds the average.

Similar plasticity rules with eligibility traces in spiking neural networks have previously been proposed by several authors \cite{Seung:03,XieSeung:04,Izhikevich:07,PfisterETAL:06,Florian:07,LegensteinETAL:08,UrbanczikSenn:09, FremauxETAL:10, FremauxETAL:13}.
\david{In \cite{FremauxETAL:13} also a method to estimate the neural firing rate $f_{\posti}(t)$ from back-propagating action potentials in the synapses has been proposed.}
The main difference to these previous approaches is that the activity-dependent last term in Eq.~\eqref{eqn:eligibility-trace} is scaled by the current synaptic weight $w_i(t)$. This weight-dependence of the update equations induces multiplicative synaptic dynamics and is a consequence of the exponential mapping Eq.~\eqref{eq:thetamap} (see derivation in the next section). This is an important property for a network model that includes rewiring. Note, that for retracted synapses ($w_i(t)=0$), both $e_{i}(t)$ and $g_i(t)$ decay to zero (within few minutes in our simulations). Therefore, we find that the dynamics of retracted synapses is only driven by the first (prior) and last (random fluctuations) term of Eq.~\eqref{eq:sde} and are independent from the network activity. Thus, retracted synapses spontaneously reappear also in the absence of reward after a random amount of time.

The first term in Eq.~\eqref{eq:sde} is the gradient of the prior distribution. We used a prior distribution that pulls the synaptic parameters towards $\thi=0$ such that unused synapses tend to disappear and new synapses are permanently formed. Throughout all simulations we used independent Gaussian priors for the synaptic parameters
\begin{equation}
  \ps{\bth} \;=\; \prod_{i} \ps{\thi}\,,\quad \text{ with } \quad \ps{\thi} \;=\; \frac{1}{\sigma\,\sqrt{2\pi}} \exp\left( -\frac{(\thi - \mu)^2}{2\sigma^2} \right)\; , \nonumber
\end{equation}
where $\sigma$ is the standard deviation of the prior distribution. Using this, we find that the contribution of the prior to the online parameter update equation is given by
\begin{equation}
  \ddthetai \log \ps{\bth} \;=\; \frac{1}{\sigma^2}\left( \mu - \thi \right)\,.
  \label{eqn:dprior-dtheta}
\end{equation}
Finally by plugging Eq.~\eqref{eqn:dprior-dtheta} and \eqref{eqn:gradient-est} into Eq.~\eqref{eq:sde} the synaptic parameter changes at time $t$ are given by
\begin{equation}
  d \thi \;=\; \beta \, \left( \frac{1}{\sigma^2}\left( \mu - \thi \right) \, + \, g_{i}(t)  \right)  dt  \;+ \; \sqrt{2 \beta T} \, d \wiener_{i} \;.
  \label{eqn:std-synapse}
\end{equation}
If not stated otherwise we used $\sigma=2$ and $\mu=0$, and a learning rate of $\beta=10^{-5}$. By inspecting Eq.~\eqref{eqn:std-synapse} it becomes immediately clear that the parameter dynamics follow an Ornstein-Uhlenbeck process it the activity-dependent second \david{term} is inactive (in the absence of reward), i.e. if $g_i(t)=0$. In this case the dynamics are given by the deterministic drift towards the mean value $\mu$ and the stochastic diffusion fueled by the Wiener process $\wiener_{i}$. The temperature $T$ and the standard deviation $\sigma$ scale the contribution of these two forces.

\subsubsection*{Reward-modulated synaptic plasticity approximates gradient ascent on the expected discounted reward}
\label{sec:gradient_expected_discounted_reward}
We first consider a theoretical setup where the network is operated in arbitrarily long episodes such that in each episode a reward sequence $\rseq$ is encountered. The reward sequence $\rseq$ can be any discrete or real-valued function that is positive and bounded. The episodic scenario is useful to derive exact batch parameter update rules, from which we will then deduce online learning rules. Due to stochastic network inputs, stochastic network responses, and stochastic reward delivery, the reward sequence $\rseq$ is stochastic.

The classical goal of reinforcement learning is to maximize the function $\mathcal{V}(\bth)$ of discounted expected rewards Eq.~\eqref{eqn:reward-prob-factorized}.
Policy gradient algorithms perform gradient ascent on $\mathcal{V}(\bth)$ by changing each parameter $\theta_i$ in the direction of the gradient $\partial \log \mathcal{V}(\bth) / \partial \theta_i$. Here, we show that the parameter dynamics Eq.~\eqref{eqn:eligibility-trace}, \eqref{eqn:gradient-est} approximate this gradient, i.e., $g_i(t) \approx \partial \log \mathcal{V}(\bth) / \partial \theta_i$.

It is natural to assume that the reward signal $\rtau$ only depends indirectly on the parameters $\bth$, through the history of network spikes $z_k(\tau)$ up to time $\tau$, which we write as $\bbztau = \{z_k(s) \;|\; 0 \leq s <\tau,\, 1\leq k \leq K  \}$, i.e., $\pn{\rt,\bbzt}{\bth} = \cprob{\rt}{\bbzt} \pn{\bbzt}{\bth}$. We can first expand the expectation $\expect[p(\rseq | \bth)]{\cdot}$ in Eq.~\eqref{eqn:reward-prob-factorized} to be taken over the joint distribution $p(\rseq, \bbz | \bth)$ over reward sequences $\rseq$ and network trajectories $\bbz$. The derivative
\begin{equation}
\ddthetai \log \mathcal{V}(\bth) \;=\; \frac{1}{\mathcal{V}(\bth)} \ddthetai \mathcal{V}(\bth) \;=\; \frac{1}{\mathcal{V}(\bth)} \ddthetai \expect[p(\rseq, \bbz | \bth)]{ \int_{0}^\infty e^{-\frac{\tau}{\tau_e}} \,\rtau \; \d \tau }
\label{eqn:der1}
\end{equation}
can be evaluated using the well-known identity $\frac{\partial}{\partial x} \expect[p( a | x)]{f(a)} = \expect[p( a | x)]{f(a) \frac{\partial}{\partial x} \log p(a | x) }$:
\begin{eqnarray}
\ddthetai \log \mathcal{V}(\bth)  &\;=\;& \frac{1}{\mathcal{V}(\bth)}  \expect[p(\rseq, \bbz | \bth)]{ \int_{0}^\infty  e^{-\frac{\tau}{\tau_e}} \, \rtau \,
                                                         \ddthetai \log \cprob{\rtau, \bbztau}{\bth}  \; \d \tau } \;
 \nonumber \\[3mm]&\;=\;& \frac{1}{\mathcal{V}(\bth)}  \expect[p(\rseq, \bbz | \bth)]{ \int_{0}^\infty  e^{-\frac{\tau}{\tau_e}} \, \rtau \,
                                                         \ddthetai\big(\log \cprob{\rtau}{\bbztau} + \log \pn{\bbztau}{\bth} \big)  \; \d \tau } \;
 \nonumber \\[3mm]
  &\;=\;&  \expect[p(\rseq, \bbz | \bth)]{ \int_{0}^\infty  e^{-\frac{\tau}{\tau_e}} \,  \frac{\rtau}{\mathcal{V}(\bth)} \, \ddthetai \log \pn{\bbztau}{\bth} \; \d \tau } \;.
\label{eqn:derivation-dlogpndtheta-1}
\end{eqnarray}
Here, $\pn{\bbztau}{\bth}$ is the probability of observing the spike train $\bbztau$ in the time interval 0 to $\tau$. For the definition of the network $\mathcal{N}$ given above, the gradient $\ddthetai \log \pn{\bbztau}{\bth}$ of this distribution can be directly evaluated.  Using Eq.~\eqref{eq:membrane-potential} and \eqref{eq:thetamap} we get \cite{PfisterETAL:06}
\begin{eqnarray}
   \ddthetai \log \pn{\bbztau}{\bth} \;&=&\; \frac{\partial w_{i}}{\partial \theta_i} \dd{w_{i}} \int_0^\tau 
     \zks \, \log \left( \fks \right) \,-\, \fks  \; \d s \nonumber \\[3mm]
     \;&\approx&\;  
     \int_0^\tau w_{i} \, \hz_{\prei}(s) \, (\zks-\fks) \, \d s \;, \label{eq:ddw_sigm}
\end{eqnarray}
where we have used that by construction only the rate function $\fks$ depends on the parameter $\theta_i$. \david{If one discretizes time and assumes that rewards and parameter updates are only realized at the end of each episode, the REINFORCE rule is recovered \cite{Williams92simple}.}

In Eq.~\eqref{eq:ddw_sigm} we used the approximation $\frac{\partial w_{i}}{\partial \theta_i} \approx w_{i}$. This expression ignores the discontinuity of Eq.~\eqref{eq:thetamap} at $\theta_i=0$, where the function is not differentiable. In practice we found that this approximation is quite accurate if $\theta_0$ is large enough such that $\exp( \theta_i - \theta_0 )$ is close to zero (which is the case for $\theta_0=3$ in our simulation). In control experiments we also used a smooth function $w_i = \exp( \theta_i - \theta_0 )$ (without the jump at $\theta_i=0$) for which Eq.~\eqref{eq:ddw_sigm} is exact, and found that this yields results that are not significantly different from the ones that use the mapping Eq.~\eqref{eq:thetamap}.

\subsubsection*{Online learning}

Eq.~\eqref{eqn:derivation-dlogpndtheta-1} defines a batch learning rule with an average taken over learning episodes where in each episode network responses and rewards are drawn according to the distribution $p(\rseq, \bbz | \bth)$. In a biological setting, there are typically no clear episodes but rather a continuous stream of network inputs and rewards and parameter updates are performed continuously (i.e., learning is online). The analysis of online policy gradient learning is far more complicated than the batch scenario, and typically only approximate results can be obtained that however perform well in practice, see e.g., \cite{Seung:03,XieSeung:04} for discussions.  

In order to arrive at an online learning rule for this scenario, we consider an estimator of Eq.~\eqref{eqn:derivation-dlogpndtheta-1} that approximates its value at each time $t>\tau_g$ based on the recent network activity and rewards during time $[t-\tau_g, t]$ for some suitable $\tau_g>0$. We denote the estimator at time $t$ by $G_i(t)$ where we want $G_i(t) \approx \ddthetai \log \mathcal{V}(\bth)$ for all $t>\tau_g$. To arrive at such an estimator,
we approximate the average over episodes in Eq.~\eqref{eqn:derivation-dlogpndtheta-1} by an average over time where each time point is treated as the start of an episode. The average is taken over a long sequence of network activity that starts at time $t$ and ends at time $t+\tau_g$. Here, one systematic difference to the batch setup is that one cannot guarantee a time-invariant distribution over initial network conditions as we did there since those will depend on the current network parameter setting. However, under the assumption that the influence of initial conditions (such as initial membrane potentials and refractory states) decays quickly compared to the time scale of the environmental dynamics, it is reasonable to assume that the induced error is negligible. We thus rewrite Eq.~\eqref{eqn:derivation-dlogpndtheta-1} in the form (we use the abbreviation $PSP_i(s) = w_{i}(s) \, \hz_{\prei}(s)$).
\begin{equation}
 \ddthetai \log \mathcal{V}(\bth) \;\approx\; G_i(t) =  \frac{1}{\tau_g} \int_{t}^{t+\tau_g} \int_{\zeta}^{t+\tau_g} \; e^{-\frac{\tau-\zeta}{\tau_e}} \;  \frac{\rtau}{\mathcal{V}(\bth)} \; \int_{\zeta}^{\tau} \, PSP_{i}(s) \, (\zks\,-\,\fks)  \; \d s \; \d \tau  \; \d \zeta  \;, \nonumber
\end{equation}
where  $\tau_g$ is the length of the sequence of network activity over which the empirical expectation is taken. Finally, we can combine the second and third integral into a single one, rearrange terms and substitute $s$ and $\tau$ so that integrals run into the past rather than the future, to obtain
\begin{equation}
  G_i(t) \;\approx\; \frac{1}{\tau_g} \int_{t-\tau_g}^{t} \; \frac{\rtau}{\mathcal{V}(\bth)} \; \int_{0}^{\tau} e^{-\frac{s}{\tau_e}} \, PSP_{i}(\tau-s) \, (\zk(\tau-s)\,-\,\fk(\tau-s))  \; \d s \; \d \tau \;,
 \label{eqn:prthetadtheta}
\end{equation}
We now discuss the relationship between $G_i(t)$ and Eq.~\eqref{eqn:eligibility-trace}, \eqref{eqn:gradient-est} to show that the latter equations approximate $G_i(t)$. Solving Eq.~\eqref{eqn:eligibility-trace} with zero initial condition $e_i(0)=0$ yields
\begin{equation}
  e_{i}(t) = \int_{0}^{t} e^{-\frac{s}{\tau_e}} \, PSP_{i}(t-s) \, (\zk(t-s)\,-\,\fk(t-s)) \, \d s  \;.
  \label{eqn:eligibility-trace-solved}
\end{equation}
This corresponds to the inner integral in Eq.~\eqref{eqn:prthetadtheta} and we can write
\begin{equation}
  G_i(t) \;\approx\; \frac{1}{\tau_g} \int_{t-\tau_g}^{t} \; \frac{\rtau}{\mathcal{V}(\bth)} \; e_i(\tau) \; \d \tau \; = \left \langle \frac{\rt}{\mathcal{V}(\bth)} \; e_i(t) \right \rangle_{\tau_g} \; \approx \left \langle \frac{\rt}{\rhatt} \; e_i(t) \right \rangle_{\tau_g},
 \label{eqn:prthetadthetarew}
\end{equation}
where $\langle \cdot \rangle_{\tau_g}$ denotes the temporal average from $t-\tau_g$ to $t$ and $\rhatt$ estimates the expected discounted reward through a slow temporal average.

Finally, we observe that any constant $\alpha$ can be added to $\rtau / \mathcal{V}(\bth)$ in Eq.~\eqref{eqn:derivation-dlogpndtheta-1} since 
\begin{equation}
\expect[p(\rseq, \bbz | \bth)]{ \int_{0}^\infty  e^{-\frac{\tau}{\tau_e}} \,  \alpha \, \ddthetai \log \pn{\bbztau}{\bth} \; \d \tau } = 0
\end{equation}
 for any constant $\alpha$ (cf. \cite{Williams92simple,UrbanczikSenn:09}).
 
Hence, we have $G_i(t) \;\approx\; \left \langle \left( \frac{\rt}{\rhatt} + \alpha \right) \; e_i(t) \right \rangle_{\tau_g}$. \david{Eq.~\eqref{eqn:gradient-est} implements this in the form of a running average and hence $g_i(t) \approx G_i(t) \approx \ddthetai \log \mathcal{V}(\bth)$ for $t>\tau_g$. Note that this result assumes that the parameters $\bth$ change slowly on the time-scale of $\tau_g$ and $\tau_g$ has to be chosen significantly longer than the time constant of the eligibility trace $\tau_e$ such that the estimator works reliably, so we require $\tau_e < \tau_g < \frac{1}{\beta}$. The time constant $\tau_a$ to estimate the average reward $\mathcal{V}(\bth)$ through $\tau_a \frac{\d \rhatt}{\d t} = -\rhatt + \rt$ should be on the same order as the time constant $\tau_g$ for estimating the gradient. We selected both to be 50~s in our simulations.
Simulations using the batch model outlined above and the online learning model showed qualitatively the same behavior for the parameters used in our experiments (data not shown).}

\subsection*{\david{Role of the prior distribution}}

\begin{figure}[ht]
\begin{center}
  \includegraphics{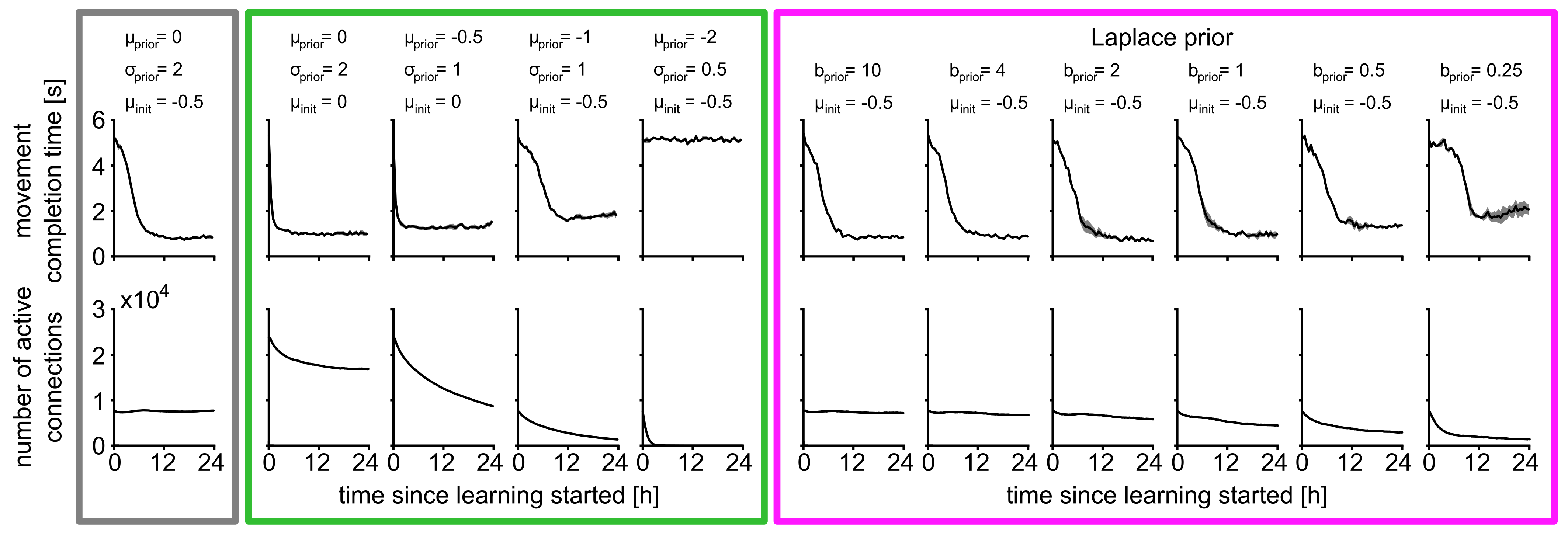}
\end{center}
\rokni{
\caption{{\bf Impact of the prior distribution on the synaptic dynamics.}
Task performance and total number of active synaptic connections throughout learning for prior distributions and distribution of initial synaptic parameters. Synaptic parameters were initially drawn from a Gaussian distribution with mean $\mu_{\text{init}}$ and $\sigma=0.5$. Gaussian prior distribution with different parameters (green) are compared to the parameter set used in all other experiments (gray). In addition a Laplace prior with different parameters was tested (purple). The prior distribution and the initial synaptic parameters had a marked effect on the task performance and overall network connectivity.
}
\label{fig:prior}
}
\end{figure}

\david{
In \figref{fig:prior} we further investigate the role of the prior distribution and initial network configuration on task performance. We measured the performance and total number of active connections for different prior distributions and distributions of initial synaptic parameters. The total number of active synaptic connections depends on the parameters of the prior distribution. We tested prior distributions with different means and variances. Stronger priors with smaller variance can be used to create sparser networks. Different parameter settings can lead to quite different network connectivities at a similar task performance. A too strong prior (e.g., $\mu=-2$, $\sigma=0.5$) leads to very sparse networks, thereby preventing learning.
}

\david{
In addition to the Gaussian prior distribution we tested a Laplace prior of the form $\ps{\theta_i}=\frac{1}{2\,b} \exp(-\frac{|\theta_i|}{b})$, with zero mean and scale parameter $b>0$. This leads to a constant negative drift term in the parameter dynamics Eq.~\eqref{eq:sde}, i.e. $\ddthetai \,\log \ps{\bth} = -\frac{1}{b}$ for active synaptic connections. Convergence to sparse solutions is faster with this prior and good task performance can be reached by networks with less active connections compared to the Gaussian prior. For example, the network with $b=2$ solved the task in 0.66 seconds on average using roughly 5700 active connections, whereas the best solution for the Gaussian prior was 0.83 seconds on average with typically more than 7500 active connections. Again, for the Laplace prior, parameters that enforced too sparse networks prevented learning.
}

\subsection*{Simulation details}

\begin{table}
\centering
\begin{tabular}{ c c l }
 \textit{symbol} &  \textit{value} &  \textit{description}\\
  \hline
  $T$ & 0.1  & temperature \\
  $\tau_e$ & 1~s  & time constant of eligibility trace \\
  $\tau_g$ & 50~s  & time constant of gradient estimator \\
  $\tau_a$ & 50~s  & time constant to estimate the average reward \\
  $\alpha$ & 0.02  & offset to reward signals \\
  $\beta$ & $10^{-5}$  & learning rate \\
  $\mu$ & 0  & mean of prior \\
  $\sigma$ & 2  & std of prior \\
\end{tabular}
\caption{\textbf{Parameters of the synapse model Eq.~\eqref{eqn:eligibility-trace}, \eqref{eqn:gradient-est} and \eqref{eqn:std-synapse}.} Parameter values were found by fitting the experimental data of \cite{YagishitaETAL:14}. If not stated otherwise, these values were used in all experiments.}
\label{tab:parameters}
\end{table}

Simulations were preformed with NEST \cite{GewaltigDiesmann:07} using an in-house implementation of the synaptic sampling model; additional tests were run in Matlab R2011b (Mathworks). The code/software described in the paper is freely available online at \cite{KappelETAL:17}. The differential equations of the neuron and synapse models were approximated using the Euler method, with fixed time steps $\Delta t = 1~\text{ms}$. All network variables were updated based on this time grid, except for the synaptic parameters $\thi$ according to Eq.~\eqref{eqn:std-synapse} which were updated only every 100~ms to reduce the computation time. Control experiments with $\Delta t =0.1~\text{ms}$, and 1~ms update steps for all synaptic parameters showed no significant differences. If not stated otherwise synaptic parameters were initially drawn from a Gaussian distribution with $\mu=-0.5$ and $\sigma=0.5$ and the temperature was set to $T=0.1$. Synaptic delays were 1~ms. Synaptic parameter changes were clipped at $\pm 4 \times 10^{-4}$ and synaptic parameters $\theta_i$ were not allowed to exceed the interval $[-2,5]$ for the sake of numerical stability.

\subsubsection*{Details to: {Task-dependent routing of synaptic connections through the interaction of stochastic spine dynamics with rewards}}

The number of potential excitatory synaptic connections between each pair of input and MSN neurons was initially drawn from a Binomial distribution ($p=0.5$, $n=10$). The connections then followed the reward-based synaptic sampling dynamics Eq.~\eqref{eq:sde} as described above. Lateral inhibitory connections were fixed and thus not subject to learning. These connections between MSN neurons were drawn from a Bernoulli distribution with $p=0.5$ and synaptic weights were drawn from a Gaussian distribution with $\mu=-1$ and $\sigma=0.2$, truncated at zero. Two subsets of ten neurons were connected to either one of the targets $T_1$ or $T_2$.

To generate the input patterns we adapted the method from \cite{KappelETAL:15}. The inputs were  representations of a simple symbolic environment, realized by Poisson spike trains that encoded sensory experiences $P_1$ or $P_2$. The 200 input neurons  were assigned to Gaussian tuning curves ($\sigma=0.2$) with centers independently and equally scattered over the unit cube. The sensory experiences $P_1$ and $P_2$ were represented by two different, randomly selected points in this 3-dimensional space. The stimulus positions were overlaid with small-amplitude jitter ($\sigma=0.05$). For each sensory experience the firing rate of an individual input neuron was given by the support of the sensory experience under the input neuron's tuning curve (maximum firing rate was 60~Hz). An additional offset of 2~Hz background noise was added. The lengths of the spike patterns were uniformly drawn from the interval [750~ms, 1500~ms]. The spike patterns were alternated with time windows (durations uniformly drawn from the interval [1000~ms, 2000~ms]) during which only background noise of 2~Hz was presented.

The network was rewarded if the assembly associated to the current sensory experience fired stronger than the other assembly. More precisely, we used a sliding window of 500~ms length to estimate the current output rate of the neural assemblies. Let $\hat{\nu}_1(t)$ and $\hat{\nu}_2(t)$ denote the estimated output rates of neural pools projecting to $T_1$ and $T_2$, respectively, at time $t$ and let $I(t)$ be a function that indicates the identity of the input pattern at time $t$, i.e. $I(t)=1$ if pattern $P_1$ was present and $I(t)=-1$ if pattern $P_2$ was present. If $I(t)(\hat{\nu}_1(t) - \hat{\nu}_2(t)) < 0$ the reward was set to $\rt=0$. Otherwise the reward signal was given by $\rt = S\left(\frac{1}{5}( I(t)\hat{\nu}_1(t) - I(t)\hat{\nu}_2(t) - \nu_0 )\right)$, where $\nu_0=25~\text{Hz}$ is a soft firing threshold and $S(\cdot)$ denotes the logistic sigmoid function. The reward was recomputed every 10~ms. During the presentation of the background patterns no reward was delivered.

In \figref[D,E]{fig:pattern-classification} we tested our reward-gated synaptic plasticity mechanism with the reward-modulated STDP pairing protocol reported in \cite{YagishitaETAL:14}. We applied the STDP protocol to 50 synapses and reported mean and s.e.m values of synaptic weight changes in \figref[D,E]{fig:pattern-classification}. Briefly, we presented 15 pre/post pairings; one per 10 seconds. In each pre/post pairing 10 presynaptic spikes were presented at a rate of 10 Hz. Each presynaptic spike was followed ($\Delta t = 10~\text{ms}$) by a brief postsynaptic burst of 3 spikes (100~Hz). The total duration of one pairing was thus 1~s indicated by the gray shaded rectangle in \figref[E]{fig:pattern-classification}. During the pairings the membrane potential was set to $u(t)=-2.4$ and Eq.~\eqref{eq:activation-function},\eqref{eqn:eligibility-trace}, \eqref{eqn:gradient-est} and \eqref{eqn:std-synapse} solved for each synapse. Reward was delivered here in the form of a rectangular-shaped wave of constant amplitude 1 and duration 300~ms to mimic puff application of dopamine. Rewards were delivered for each pre/post pairing and reward delays were relative to the onset of the STDP pairings. \david{The time constants $\tau_e$ and $\tau_g$, the reward offset $\alpha$ and the temperature $T$ of the synapse model were chosen to qualitatively match the results of Fig.~1 and Fig.~4 of \cite{YagishitaETAL:14} (see Tab.~\ref{tab:parameters}). The value of $\tau_a$ for the estimation of the average reward has been chosen to equal $\tau_g$ based on theoretical considerations (see \emph{Online learning}). We found that the parameters of the prior had relatively small effect on the synaptic dynamics on timescales of one hour. We tuned the parameters of the prior distribution by hand to achieve good results on the task presented in \figref{fig:exp-peters} (see \figref{fig:prior} for a comparison of different prior distributions).} These parameters were used in all experiments if not stated otherwise.

Synaptic parameter changes in \figref[G]{fig:pattern-classification} were measured by taking snapshots of the synaptic parameter vectors every 4~minutes. Parameter changes were measured in terms of the Euclidean norm of the difference between two successively recorded vectors. The values were then normalized by the maximum value of the whole experiment and averages over 5 trials were reported.

\rokni{
To generate the dPCA projection of the synaptic parameters in \figref[J]{fig:pattern-classification}, we adopted the methods of \cite{KobakETAL:16}. We randomly selected a subset of 500 synaptic parameters to compute the projection. We sorted the parameter vectors by the average reward achieved over a time window of 10 minutes and binned them into 10 equally spaced bins. The dPCA algorithm was then applied on this dataset to yield the projection matrix and estimated fractions of explained reward variance. The projection matrix was then applied to the whole trajectory of network parameters and the first 2 components were plotted. The trajectory was projected onto the estimated expected reward surface based on the binned parameter vectors.
}

\subsubsection*{Details to: {A model for task-dependent self-configuration of a recurrent network of excitatory and inhibitory spiking neurons}}
\begin{figure}[ht]
\begin{center}
  \includegraphics{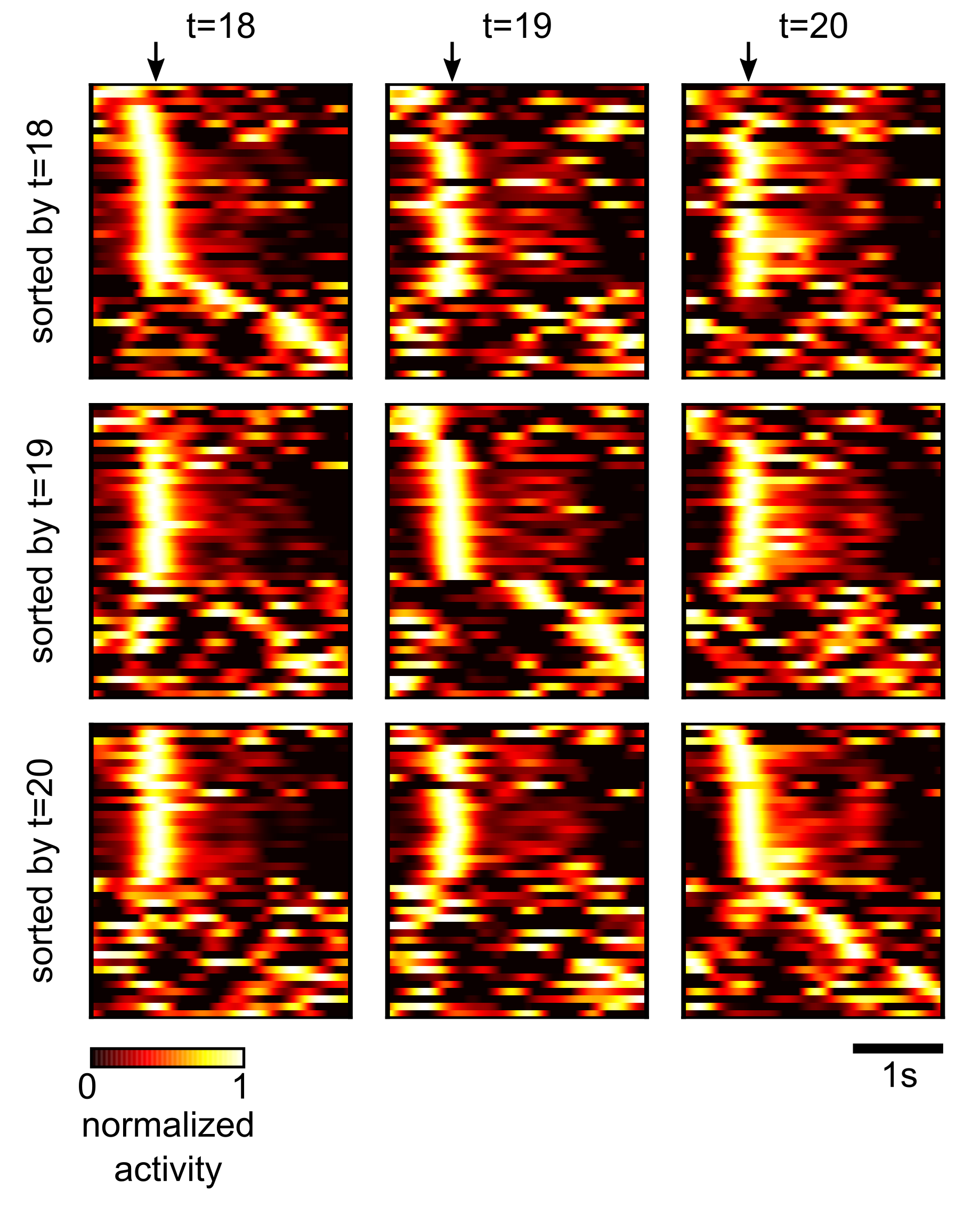}
\end{center}
\caption{{\bf Drifts of neural codes while performance remained constant.}
Trial-averaged network activity as in \figref[D]{fig:exp-peters} evaluated at three different times selected from  a time window where the network performance was stable (see \figref[C]{fig:exp-peters}). Each column shows the same trial-averaged activity plot but subject to different sorting. Rows correspond to one sorting criterion based on one evaluation time.
}
\label{fig:peths}
\end{figure}

Neuron and synapse parameters were as reported above, except for the inhibitory neurons for which we used faster dynamics with a refractory time $t_\text{ref} = 2~\text{ms}$ and time constants $\tau_m = 10~\text{ms}$ and $\tau_r = 1~\text{ms}$ for the PSP kernel. The network connectivity between excitatory and inhibitory neurons was as suggested in \cite{AvermannETAL:12}. Excitatory (pools D, U and hidden) and inhibitory neurons were randomly connected with connection probabilities  given in Table~2 in \cite{AvermannETAL:12}. Connections include lateral inhibition between excitatory and inhibitory neurons. The connectivity to and from inhibitory neurons was kept fixed throughout the simulation (not subject to synaptic plasticity or rewiring). The connection probability from excitatory to inhibitory neurons was given by 0.575. The synaptic weights were drawn from a Gaussian distribution (truncated at zero) with $\mu = 0.5$ and $\sigma = 0.1$. Inhibitory neurons were connected to their targets with probability 0.6 (to excitatory neurons) and 0.55 (to inhibitory neurons) and the synaptic weights were drawn from a truncated normal distribution with $\mu = -1$ and $\sigma = 0.2$. The number of potential excitatory synaptic connections between each pair of excitatory neurons was drawn from a Binomial distribution ($p=0.5$, $n=10$). These connections were subject to the reward-based synaptic sampling and rewiring described above. \david{In the resulting network scaffold around $49\%$ of connections consisted of multiple synapses.}

To infer the lever position from the network activity, we weighted spikes from the neuron pool D with $-1$ and spikes from U with $+1$, summed them and then filtered them with a long PSP kernel with $\tau_r = 50~\text{ms}$ (rise) and $\tau_m = 500~\text{ms}$ (decay). The cue input pattern was realized by the same method that was used to generate the patterns $P_1$ and $P_2$ outlined above. If a trial was completed successfully the reward signal $\rt$ was set to 1 for 400~ms and was 0 otherwise. After each trial a short holding phase was inserted during which the input neurons were set to 2~Hz background noise. The lengths of these holding phases were uniformly drawn from the interval [1~s, 2~s]. At the time points marked by $\ast$, the reward policy was changed by switching the decoding functions of the neural pools D and U and by randomly re-generating the input cue pattern.

To identify the movement onset times in \figref[D]{fig:exp-peters} we adapted the method from \cite{PetersETAL:14}. Lever movements were recorded at a sampling rate of 5~ms. Lever velocities were estimated by taking the difference between subsequent time steps and filtering them with a moving average filter of 5 time steps length. A Hilbert transform was applied to compute the envelope of the lever velocities. The movement onset time for each trial was then defined as the time point where the estimated lever velocity exceeded a threshold of 1.5 in the upward movement direction. If this value was never reached throughout the whole trial the time point of maximum velocity was used (most cases at learning onset).

The trial-averaged activity traces in \figref[D]{fig:exp-peters} were generated by filtering the spiking activity of the network with a Gaussian kernel with $\sigma = 75~\text{ms}$. The activity traces were aligned with the movement onset times (indicated by black arrows in \figref[D]{fig:exp-peters}) and averaged across 100 trials. The resulting activity traces were then normalized by the neuron's mean activity over all trials and values below the mean were clipped. The resulting activity traces were normalized to the unit interval.

Turnover statistics of synaptic connections in \figref[E]{fig:exp-peters} were measured as follows. The synaptic parameters were recorded in intervals of 2 hours. The number of synapses that appeared (crossed the threshold of $\thi=0$ from below) or disappeared (crossed $\thi=0$ from above) between two measurements were counted and the total number was reported as turnover rate.

\rokni{
For the consolidation mechanism in \figref[F]{fig:exp-peters} we used a modified version of the algorithm where we introduced for each synaptic parameter $\theta_i$ an independent mean $\mu_i$ for the prior distribution $\ps{\bth}$. After 4 simulated days we set $\mu_i$ to the current value of $\theta_i$ for each synaptic parameter and the standard deviation $\sigma$ was set to 0.05. Simulation of the synaptic parameter dynamics was then continued for 10 subsequent days.
}

For the approximation of simulating retracted potential synaptic connections in \figref[C,G]{fig:exp-peters} we paused evaluation of the SDE \eqref{eq:sde} for $\theta_i \leq 0$. Instead, synaptic parameters of retracted connections where randomly set to values above zero after random waiting times drawn from an Exponential distribution with a mean of 12 hours. When a connection became functional at time $t$ we set $\theta_i(t) = 10^{-5}$ and reset the eligibility trace $e_i(t)$ and gradient estimator $g_i(t)$ to zero and then continued the synaptic dynamics according to \eqref{eq:sde}. Histograms in \figref[F]{fig:exp-peters} were computed over bins of 2 hours width.

In \figref{fig:peths} we further analyzed the trial-averaged activity at three different time points (18~h, 19~h and 20~h) where the performance was stable (see \figref[C]{fig:exp-peters}). Drifts of neural codes on fast time scales could also be observed during this phase of the experiment.

\begin{figure}[ht]
\begin{center}
  \includegraphics{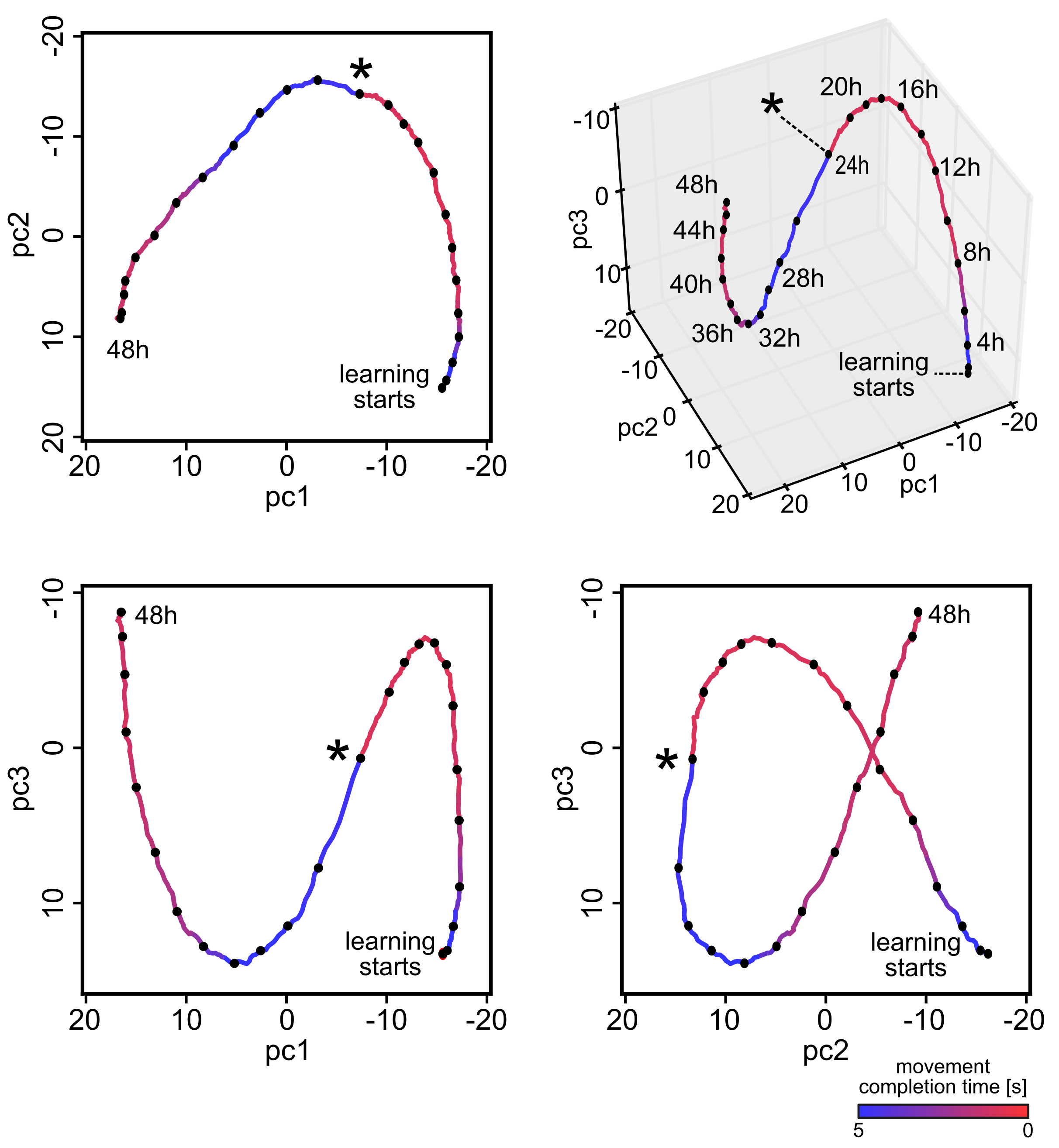}
\end{center}
\caption{\david{{\bf 2D projections of the PCA analysis in \figref[I]{fig:exp-peters}.} The 3D projection as in \figref[I]{fig:exp-peters} (top right) and the corresponding 2D projections are shown.}
}
\label{fig:pca-proj}
\end{figure}

\david{
Since the 3D illustration of the PCA projection in \figref[I]{fig:exp-peters} is ambiguous, the corresponding 2D projections are shown in \figref{fig:pca-proj}. The projection to the first 2 components (pc1 and pc2) show the migration of synaptic parameters to a new region after the task change. The first 3 principal components explain $82\%$ of the total variance in the parameter dynamics.
}

\subsubsection*{Details to: {Compensation for network perturbations}}

The black curve in \figref[C]{fig:exp-peters} shows the learning curve of a network for which rewiring was disabled after the task change at 24~h. Here, synaptic parameters were not allowed to cross the threshold at $\theta_i = 0$ and thus could not change sign after 24~h. Apart from this modification the synaptic dynamics evolved according to Eq.~\eqref{eqn:std-synapse} as above with $T=0.1$.

For the analysis of synaptic turnover in \figref[G]{fig:exp-peters} we recorded the synaptic parameters at $t_1=24$~h and $t_2=48$~h. We then classified each potential synaptic connection $i$ into one of four classes, stable non-functional: $(\theta_i(t_1) \leq 0) \wedge (\theta_i(t_2) \leq 0)$, transient decaying: $(\theta_i(t_1) > 0) \wedge (\theta_i(t_2) \leq 0)$, transient emerging: $(\theta_i(t_1) \leq 0) \wedge (\theta_i(t_2) > 0)$ and stable functional: $(\theta_i(t_1) > 0) \wedge (\theta_i(t_2) > 0)$.

In \figref[H]{fig:exp-peters} we randomly selected 5$\%$ of the synaptic parameters $\theta_i$ and recorded their traces over a learning experiment of 48~hours (1 sample per minute). The principal component analysis (PCA) was then computed over these traces, treating the parameter vectors at each time point as one data sample. The high-dimensional trace was then projected to the first three principal components in \figref[H]{fig:exp-peters}, and colored according to the average movement completion time that was acquired by the network at the corresponding time points.

\subsubsection*{Details to: {Relative contribution of spontaneous and activity-dependent processes to synaptic plasticity}}
Synaptic weights in \figref[a,b]{fig:ci-vs-no-ci} were recorded in intervals of 10 minutes. We selected all pairs of synapses with common pre- and postsynaptic neurons as CI synapses and synapse pairs with the same post- but not the same presynaptic neuron as non-CI synapses. In \figref[d-f]{fig:ci-vs-no-ci} we took a snapshot of the synaptic weights after 48 hours of learning and computed the Pearson correlation of all CI and non-CI pairs for random subsets of around 5000 pairs. Data for 100 randomly chosen CI synapse pairs are plotted of \figref[E]{fig:ci-vs-no-ci}.

In \figref[F]{fig:ci-vs-no-ci} we analyzed the contribution of activity-dependent and spontaneous processes in our model. \cite{DvorkinZiv:16} reported that a certain degree of the stochasticity in their results could be attributed to their experimental setup. The maximum detectable correlation coefficient was limited to $0.76-0.78$, due to the variability of light fluorescence intensities which were used to estimate the sizes of postsynaptic densities. Since in our computer simulations we could directly read out values of the synaptic parameters we were not required to correct our results for noise sources in the experimental procedure (see p.~16ff and equations on p.~18 of \cite{DvorkinZiv:16}). This is also reflected in our data by the fact that we got a correlation coefficient that was close to 1.0 in the case $T=0$ (see \figref[D]{fig:ci-vs-no-ci}). Following the procedure of \cite{DvorkinZiv:16} we estimated in our model the contributions of activity history dependent and spontaneous synapse-autonomous processes as in Fig.~8E of \cite{DvorkinZiv:16}. Using the assumption of zero measurement error and thus a theoretically achievable maximum correlation coefficient of $r=1.0$. \david{The Pearson correlation of CI-synapses was given by $0.46 \pm 0.034$ and that of non-CI synapses by $0.08 \pm 0.015$. Therefore, we estimated the fraction of contributions of specific activity histories to synaptic changes (for $T=0.15$) as $0.46 - 0.08 = 0.38$ and of spontaneous synapse-autonomous processes as $1.0 - 0.46 = 0.54$ \cite{DvorkinZiv:16}. The remaining $8\%$ (measured correlation between non-CI synapses) resulted from processes that were not specific to presynaptic input, but specific to the activity of the postsynaptic neuron (neuron-wide precesses).}

\subsection*{Acknowledgments}
Written under partial support by the Human Brain Project of the European Union $\#$604102 and $\#$720270.


\end{document}